\documentclass[12pt]{JHEP3}

\usepackage{amsmath}
\usepackage{amssymb}
\usepackage{amsthm}

\newcommand{\RR}{\mathbb{R}} 
\newcommand{\QQ}{\mathbb{Q}} 
\newcommand{\FF}{\mathbb{F}} 
\newcommand{\ZZ}{\mathbb{Z}}

\newcommand{\MM}{\mathbb{M}}
\newcommand{\Hh}{\mathcal{H}}

\newcommand{\M}{\mathcal{M}}
\newcommand{\N}{\mathcal{N}}

\newcommand{\calS}{\mathcal{S}}
\newcommand{\TT}{\mathbb{T}}
\newcommand{\T}{\mathcal{T}}
\newcommand{\R}{\mathcal{R}}
\newcommand{\I}{\mathcal{I}}
\newcommand{\C}{\mathcal{C}}
\newcommand{\Q}{\mathcal{Q}}

\newcommand{\rrangle}{\rangle\!\!\!\;\rangle}
\newcommand{\llangle}{\langle\!\!\!\;\langle}

\newtheorem{theorem}{Theorem}
\newtheorem{proposition}[theorem]{Proposition}

\DeclareMathOperator{\Tr}{Tr}
\DeclareMathOperator{\lcm}{lcm}
\DeclareMathOperator{\rk}{rk}

\numberwithin{equation}{section}

\def\be{\begin{equation}}
\def\ee{\end{equation}}

\allowdisplaybreaks[3]

\title{\begin{flushright}{\vspace{-0.8cm}\small MPP-2011-69}\end{flushright}
\vspace{0.5cm}Symmetries of K3 sigma models}

\author{Matthias R.~Gaberdiel$^1$, Stefan Hohenegger$^2$ and
Roberto Volpato$^1$
\\ ~
\\
${}^{1}$ Institut f\"ur Theoretische Physik \\
\;\; ETH Zurich \\
\;\; CH-8093 Z\"urich \\
\;\; Switzerland  \\
\\
${}^{2}$ Max-Planck-Institute for Physics  \\
\;\; F\"ohringer Ring 6  \\
\;\; D-80805 M\"unchen  \\
\;\; Germany}

\abstract{It is shown that the supersymmetry-preserving automorphisms of any non-linear $\sigma$-model
on K3 generate a subgroup of the Conway group $Co_1$. This is the stringy generalisation of the classical
theorem, due to Mukai and Kondo, showing that the symplectic automorphisms of any K3 manifold
form a subgroup of the Mathieu group $\MM_{23}$. The Conway  group $Co_1$ contains
the Mathieu group $\MM_{24}$ (and therefore in particular  $\MM_{23}$) as a subgroup.
We confirm the predictions of the Theorem with three explicit CFT realisations of K3:
the ${\mathbb T}^4/\ZZ_2$ orbifold at the self-dual point, and the two Gepner models $(2)^4$ and 
$(1)^6$. In each case we demonstrate that their symmetries do not form a subgroup of $\MM_{24}$,
but lie inside $Co_1$ as predicted by our Theorem.}

\begin{document}

\section{Introduction}

Recently, a hidden $\MM_{24}$ symmetry of the elliptic genus of  K3 has attracted some
attention. This development started with the observation of Eguchi, Ooguri and
Tachikawa \cite{EOT} who noted that the first few multiplicities with which the ${\cal N}=4$ characters 
appear in the elliptic genus of K3 are sums (with integer coefficients) of dimensions of representations of the
Mathieu group $\MM_{24}$. The appearance of these dimensions suggests that the underlying 
vector space (consisting of the states that contribute to the elliptic genus) carries an action of
$\MM_{24}$. Assuming this group action one can then also define the `twining genera'\footnote{These
are the analogues of the so-called MacKay-Thompson series for Monstrous Moonshine, see
\cite{Gannon} for a modern review.}, {\it i.e.}\ the elliptic genus with the insertion of a group
element $g\in \MM_{24}$,
\be
\Tr_{\rm RR} \Bigl( g \, y^{J_0} q^{L_0-\frac{c}{24}} (-1)^F\, \bar{q}^{\bar{L}_0-\frac{c}{24}} (-1)^{\bar{F}} \Bigr) \ ,
\ee
and the usual string arguments suggest that these twining 
genera must have good modular properties under some congruence subgroup of $SL(2,\mathbb{Z})$. 
Using these modular properties as well as the explicit knowledge of the first few coefficients that follow
 from the observation of \cite{EOT}, all twining genera could be determined
\cite{Cheng:2010pq,Gaberdiel:2010ch,Gaberdiel:2010ca,Eguchi:2010fg}. 
In turn, this leads to a stringent test of the proposal:
knowing all  twining genera one can deduce the decomposition of all
multiplicity spaces into $\MM_{24}$ representations, and it was found that at least for the
first 1000 coefficients, non-negative integer multiplicities appear 
\cite{Gaberdiel:2010ca,Eguchi:2010fg}.\footnote{We thank
Yuji Tachikawa for informing us that he has now checked the decomposition for the first $1000$
coefficients.} This analysis therefore gives very convincing
evidence for a hidden $\MM_{24}$ symmetry underlying the elliptic genus of K3. Further
support for this conjecture was given in \cite{Govindarajan:2010cf} (see also 
\cite{Govindarajan:2010fu,Govindarajan:2009qt}).

Part of this symmetry can be understood geometrically. First of all, the elliptic genus
is independent of the specific point in the moduli space of K3 that is considered, and
thus the symmetries of the elliptic genus are in some sense the union of all 
symmetries that are present at different points in moduli space. The geometrical symmetries
of K3,  {\it i.e.}\ the symplectic automorphisms, have been studied some time ago by 
Mukai and Kondo \cite{Mukai,Kondo}, and they found that at any point in moduli space
these symmetries form a subgroup of $\MM_{23}$. The Mathieu group $\MM_{23}$ is a 
maximal subgroup of $\MM_{24}$, and thus this argument `explains' part of the observation of 
\cite{EOT}.

As is familiar for example from T-duality, string theory typically has more than just
the geometric symmetries, and one may therefore expect that
the remaining symmetries of $\MM_{24}$ may be accounted for by `stringy symmetries' 
(see also \cite{Taormina:2010pf}). In order
to analyse this question, we study in this paper the stringy version of the Mukai-Kondo theorem. 
More specifically, we classify the spacetime supersymmetry preserving automorphisms of the
non-linear $\sigma$-model at an arbitrary point in the moduli space of K3.
From the point of view of the worldsheet, these symmetries are characterised by the property that they
preserve the ${\cal N}=(4,4)$ superconformal algebra, as well as the  spectral flow operators.
Given the observation of \cite{EOT}, one may have expected that all these 
symmetries should form a subgroup of $\MM_{24}$, but actually the answer is more complicated:
\smallskip

\noindent {\bf Theorem:}
{\it Let $G$ be the group of symmetries of a non-linear $\sigma$-model on $K3$ preserving 
the $\N=(4,4)$ superconformal algebra as well as the spectral flow operators. 
Then one of the following possibilities holds:
\begin{list}{{\rm (\roman{enumi})}}{\usecounter{enumi}}
\item $G=G'. G''$, where $G'$ is a subgroup of $\ZZ_2^{11}$, and $G''$ is a subgroup of  
$\MM_{24}$ with at least four orbits when acting as a permutation on
$\{1,\ldots,24\}$
\item $G=5^{1+2}.\ZZ_4$
\item $G=\ZZ_3^4.A_6$
\item $G=3^{1+4}.\ZZ_2.G''$, where $G''$ is either trivial, $\ZZ_2$, $\ZZ_2^2$ or $\ZZ_4$.
\end{list}}
\smallskip

\noindent Here $p^{1+2n}$ denotes an extra special group of order $p^{1+2n}$,
 and $N.Q$ denotes a group $G$ for which $N$ is a normal subgroup such that 
 $G/N\cong Q$ (for an exposition of our mathematical notation and conventions see 
 Appendix~\ref{s:notation}). Note that except for case (i) with $G'$ trivial, these groups 
 are not subgroups of $\MM_{24}$; in particular, for cases (ii)-(iv) this follows from the fact 
 that their order does not divide 
\be 
|\MM_{24}|=2^{10}\cdot 3^3\cdot 5\cdot 7\cdot 11\cdot23\ .
\ee
On the other hand, all groups in (i)--(iv) {\em are} subgroups of the Conway group 
$Co_1$,  and thus the analogue of the Mukai-Theorem is that the stringy symmetries all
lie in $Co_1$. One may take this as evidence that the elliptic genus of K3
should in fact have a hidden $Co_1$ symmetry, but from its decomposition in terms of 
$\N=4$ elliptic genera, we have not seen any hint for this. In any case, the result of the Theorem
means that the explanation of the $\MM_{24}$ symmetry appearing in the elliptic genus of K3 
must be more subtle. 

\subsection{Sketch of Proof}

Let us briefly sketch the proof of the Theorem, before returning to more general considerations below;
the details of this argument as well as the underlying assumptions will be spelled out in 
Section~\ref{symmetries} and Appendix~\ref{s:latticeproofs}.

The basic strategy of the proof follows closely the proof of the Mukai-theorem, given by Kondo.
The moduli space of sigma-models on K3 has the form
\be 
\M_{\rm K3} = O(\Gamma^{4,20})\backslash O(4,20)/(O(4)\times O(20))\ .
\ee
Here the Grassmannian $O(4,20)/(O(4)\times O(20))$ parametrises the choice of a positive
definite four-dimensional subspace $\Pi\subset \RR^{4,20}$, and $O(\Gamma^{4,20})$ is the group of
automorphisms of the even unimodular lattice $\Gamma^{4,20}\subset \RR^{4,20}$ of signature $(4,20)$. 
We may think of $\Gamma^{4,20}$ as the integral homology of K3, {\it i.e.}\ as the 
D-brane charge lattice, while the position of $\Pi$ is specified by the choice of a Ricci-flat metric
and a $B$-field on K3. In particular, it therefore determines the four left- and right-moving supercharges.

The supersymmetry preserving automorphisms of the non-linear $\sigma$-model
characterised by $\Pi$ generate the group $G\equiv G_\Pi$ that consists of 
those elements of $O(\Gamma^{4,20})$ 
that leave $\Pi$ pointwise fixed.  We denote by $L^G$ the sublattice of $G$-invariant vectors of 
$L\equiv \Gamma^{4,20}$, and define $L_G$ to be its orthogonal complement. 
By construction, $\Pi$ is a subspace of the real vector space $L^G\otimes \RR\subset \RR^{4,20}$,
and since $\Pi$ has signature $(4,0)$, the orthogonal complement $L_G$ must be a negative definite
lattice of rank at most $20$. The basic idea is now to embed $L_G(-1)$ --- the $(-1)$ means
that we change the sign of its intersection matrix --- into the Leech lattice $\Lambda$. Such an embedding 
exists, provided we assume that $L_G(-1)$ does not contain any vectors of length squared two
(which would signal some gauge enhancement and thus would lead to a singular CFT). Since the 
action of $G$ fixes all vectors of $\Lambda$ orthogonal 
to $L_G(-1)$, it follows that $G$ must be a subgroup of $Co_1 \subset Co_0={\rm Aut}(\Lambda)$ that 
fixes pointwise a sublattice of the Leech lattice of rank at least $4$. A more careful 
analysis then leads to the separate cases (i)-(iv) above.

\subsection{Comments and Outline}

Since the result of the Theorem is somewhat contrary to expectations, we have also studied
a few explicit conformal field theories describing K3 at different points in moduli space in detail. 
In particular, we have done this for (A) the orbifold point ${\mathbb T}^4/{\mathbb Z}_2$;
(B) the orbifold point ${\mathbb T}^4/{\mathbb Z}_4$ which is equivalent to the Gepner model
$(2)^4$; and (C) the Gepner model $(1)^6$. Given that these descriptions are very explicit, it
is possible to identify (at least some of) the supersymmetry-preserving automorphisms.
In each case we have computed the resulting symmetry group, and compared it with the 
possibilities allowed for by the Theorem. We find that (C) realises case (iii), while 
both (A) and (B) correspond to case (i) with $G'$ non-trivial. In particular, all of these
cases therefore describe K3s for which the stringy symmetries do not lie inside 
$\MM_{24}$. 

{}From the point of view of the argument leading to the Theorem, the symmetries of the 
worldsheet theory should have an interpretation as lattice symmetries. Actually, this
point of view can also be directly understood in conformal field theory: as mentioned before, 
the lattice $\Gamma^{4,20}$ can be identified with the D-brane charge lattice, and $\Pi$ 
describes the four left- and right-moving supercharges. Thus the supersymmetry preserving 
automorphisms
of the CFT should be in one-to-one correspondence with the symmetries of the D-brane charge 
lattice that leave $\Pi$ pointwise invariant. Using conformal field theory methods, it is 
fairly straightforward to determine the D-brane charge lattice, as well as the $\Pi$-preserving 
symmetries. For each of the three cases (A)--(C) we have verified that 
the resulting symmetry groups reproduce precisely those obtained from the explicit construction
of the symmetry generators above.  Incidentally, this method to determine the 
supersymmetry preserving automorphisms by analysing the D-brane charge lattice constitutes
a nice general approach that can be 
applied to any non-linear $\sigma$-model on K3.
\bigskip

The paper is organised as follows. In the following section (Section~\ref{symmetries}) 
we give a more detailed description of the main Theorem and the assumptions that go into its 
proof. Section~\ref{s:orbifold} is devoted to the study of the orbifold point ${\mathbb T}^4/{\mathbb Z}_2$.
Among other things, we calculate  the twining genera for the various symmetries
and find the twining genus of the 2B conjugacy class of $\MM_{24}$ that does not lie inside
$\MM_{23}$. (The corresponding symmetry is the stringy 4-fold T-duality symmetry $\T$, see
(\ref{Z2T}).) However,
we also find symmetries that do not lie inside $\MM_{24}$, and whose twining genus 
does not agree with the twining genus of any conjugacy class in $\MM_{24}$. (The simplest
example is the `quantum symmetry' $\Q$, see (\ref{Z2Q}).) In Section~\ref{s:16}, the 
same analysis is done for the $(1)^6$ Gepner point, while Section~\ref{s:24} deals with
the $(2)^4$ Gepner  model. Our notation and some basic mathematical background is
described in Appendix~\ref{s:notation}, while some of the details of the proof of the Theorem have
been delegated to Appendix~\ref{s:latticeproofs}.  Finally, Appendix~\ref{s:Gepner} contains
some of the details of the D-brane charge analysis for the Gepner models.

\section{Symmetries of Non-linear $\sigma$-models on K3}\label{symmetries}

In the following we shall consider two-dimensional theories with $\N=(4,4)$ superconformal
symmetry and central charge $c=6$. Theories of this type can be classified according to their
elliptic genus $\phi(\tau,z)$. The symmetries of the theory constrain the elliptic genus to be a
weak Jacobi form of weight $0$ and index $1$, and the only possibilities are $\phi(\tau,z)=0$, 
which corresponds to the case of the target space being ${\mathbb T}^4$, or
\be\label{K3eg}
\phi(\tau,z)=8\sum_{i=2}^4\frac{\vartheta_i(\tau,z)^2}{\vartheta_i(\tau,0)^2}=2y+20+2y^{-1}+O(q)\ ,
\ee
where $\vartheta_i$ are the Jacobi theta functions. This second case arises if the target space is K3
and it is the main focus of the present paper.

\subsection{Moduli Space of Non-linear  $\sigma$-models on K3}

As mentioned in the introduction, the 
moduli space of $\N=(4,4)$ theories with elliptic genus \eqref{K3eg} is believed to be the quotient
(see for example \cite{Aspinwall:1996mn,NahmWend})
\be \label{K3moduli}
\M_{\rm K3} = O(\Gamma^{4,20})\backslash O(4,20)/(O(4)\times O(20))\ .
\ee
Here the Grassmannian $O(4,20)/(O(4)\times O(20))$ parametrises the choice of a positive
definite four-dimensional subspace in $\RR^{4,20}$, and $O(\Gamma^{4,20})$ is the group of
automorphisms of the even unimodular lattice $\Gamma^{4,20}$ with signature $(4,20)$. 

Geometrically, we can think of $\Gamma^{4,20}$ as the integral homology lattice 
$H_{\rm even}(X,\ZZ)$ of the K3 manifold $X$, with the bilinear form given by the intersection number. 
The space $\RR^{4,20}$ is interpreted as the real even
cohomology $H^{\rm even}(X,\RR)$ endowed with the cup product, and the embedding
$\Gamma^{4,20}\subset \RR^{4,20}$ is realised through Poincar\'e duality
$H_{\rm even}(X,\ZZ)\cong H^{\rm even}(X,\ZZ)\subset H^{\rm even}(X,\RR)$. The non-linear
$\sigma$-model is determined by choosing a Ricci-flat metric and a B-field on the manifold $X$;
this corresponds to the choice of the $4$-dimensional subspace $\Pi\subset H^{\rm even}(X,\RR)$,
though the relationship is rather involved (see \emph{e.g.}~\cite{Aspinwall:1996mn}).

In string theory, the homology lattice can be identified with the lattice of 
D-brane charges, and the intersection number for $\alpha,\beta\in H_{\rm even}(X,\ZZ)$
is reproduced by the overlap
\be \label{inter}
\alpha\cdot\beta = \llangle \alpha \| q^{\frac{1}{2}(L_0+\tilde L_0)-\frac{c}{12}}(-1)^{F_L}
\|\beta\rrangle_{\rm RR}\ .
\ee
Here only the RR part of the boundary states $\|\alpha\rrangle,\|\beta\rrangle$ corresponding to the
D-branes wrapping the cycles $\alpha$ and $\beta$ contribute \cite{Brunner:1999jq}. (Alternatively, 
(\ref{inter}) is the Witten index in the R-sector of the relative open string.) In this picture, the dual space of 
real cohomology is naturally identified with the space of $24$ (anti-)chiral RR 
ground states with $h=\bar h=\tfrac{1}{4}$.
Under the action of ${\rm SU}(2)_L\times {\rm SU}(2)_R$, which is part of the $\N=(4,4)$ 
superconformal symmetry, the $24$ RR ground states split into a four-dimensional 
$({\bf 2},{\bf 2})$-representation and $20$ singlets. The 
four-dimensional subspace $\Pi\subset \RR^{4,20}$ is then to be identified with the
subspace of RR states transforming in the $({\bf 2},{\bf 2})$-representation.
\medskip

\subsection{Characterisation of $\N=(4,4)$ Preserving Symmetries}

Our goal is to classify the discrete symmetries $\tilde G_\Pi$ of a given $\N=(4,4)$ theory 
parametrised by $\Pi$ in the 
moduli space $\M_{\rm K3}$. Any symmetry $g\in \tilde{G}_\Pi$ must obviously leave $\Pi$
invariant. It must therefore either be an element of 
$O(\Gamma^{4,20})$, or it must act trivially on the Grassmannian in (\ref{K3moduli}). However,
the latter case would  imply that the symmetry exists everywhere in moduli space,
and we know (for example from studying deformations of the orbifold line) that this is not possible.
Thus we conclude that the symmetries $\tilde{G}_\Pi$ of the theory
at $\Pi$ is precisely the subgroup of $O(\Gamma^{4,20})\subset O(4,20,\RR)$ that 
leaves $\Pi$ (setwise) fixed.
\medskip

A general symmetry $g\in \tilde G_\Pi$ will preserve the
$\N=(4,4)$ superconformal algebra only up to an automorphism. From now on we want to focus
on the subgroup $G_\Pi\subset \tilde G_\Pi$ that actually leaves the $\N=(4,4)$ superconformal algebra 
invariant, {\it i.e.}\ for which this automorphism is trivial, and that preserve the spectral flow operators.
These symmetries are characterised by the condition that they preserve space-time
supersymmetry \cite{Banks:1988yz}. They are also relevant for the analysis of
`Mathieu Moonshine'  \cite{EOT} that was reviewed at the beginning of the Introduction.
Indeed, $\MM_{24}$ appears to act on the multiplicity spaces with which the
${\cal N}=4$ representations appear in the elliptic genus, and hence must commute with
the left-moving $\N=4$ superconformal algebra. Furthermore, according to the proposal of
\cite{Cheng:2010pq,Gaberdiel:2010ch}, the four RR ground states that transform in the $({\bf 2},{\bf 2})$ of the
${\rm SU}(2)_L\times {\rm SU}(2)_R$ --- these are the spacetime supercharges ---
sit in a singlet representation of $\MM_{24}$. Thus the symmetries that are described by
$\MM_{24}$ should leave the full $\N=(4,4)$ superconformal algebra invariant, and
preserve the spectral flow operators. Obviously, 
$\MM_{24}$ is not a subgroup of $O(\Gamma^{4,20})$, and thus
we cannot explain the full $\MM_{24}$ symmetry by looking at just one point in moduli space. However, 
the elliptic genus is constant over moduli space, and one may therefore expect that
we can account for the entire $\MM_{24}$ by putting information from different points in
moduli space together. This is one of the main motivations for classifying the
symmetry groups at different points in moduli space.

In any case, since  $\Pi\subset \RR^{4,20}$ can be identified with the
subspace of RR states transforming in the $({\bf 2},{\bf 2})$-representation, we conclude that

\begin{verse}\label{th:generic} 
\noindent {\it The subgroup $G_\Pi$ of symmetries of the K3 $\sigma$-model 
characterised by $\Pi\in\RR^{4,20}$ that leave the $\N=(4,4)$ superconformal algebra 
invariant and that preserve the spectral flow operators,
is the subgroup of $O(\Gamma^{4,20})\subset O(4,20,\RR)$ that leaves 
$\Pi$ (pointwise) fixed.}
\end{verse}

In the following, by a symmetry of an $\N=(4,4)$ theory in $\M_{\rm K3}$, we will always mean a
transformation with these properties. Note that this restriction excludes some very interesting symmetries,
for example mirror symmetry in a self-mirror theory. On the other hand, the extension of our arguments
to more general cases is fairly straightforward.

\subsection{Classification of the Groups of Symmetries}\label{s:classif}

In this section we will classify the possible groups $G_\Pi$. 
The related problem in classical geometry has been previously considered
by Mukai \cite{Mukai}, who classified the groups of symplectic automorphisms of
K3 surfaces. Mukai proved that the symplectic automorphisms of any K3 surface form a subgroup of
$\MM_{23}$, which in turn is a maximal subgroup of $\MM_{24}$. The Mathieu group
$\MM_{23}$ is finite and its subgroups are well studied, so that the Mukai theorem provides a very
explicit description of all symplectic automorphisms.
In the following, we will extend the Mukai theorem to the classification of the symmetry groups
of the $\sigma$-models. In particular, we will show that $G_\Pi\subset Co_1$, the 
Conway group $Co_1$. Note that $Co_1$ contains $\MM_{23}$ (as well as $\MM_{24}$) as a subgroup.

As argued in the previous section, the  symmetries of interest form the
subgroup of $O(\Gamma^{4,20})\subset O(4,20,\RR)$ that fix (pointwise) the positive-definite
four-dimensional subspace $\Pi\subset \RR^{4,20}$, characterising the relevant point in moduli space.
However, not all choices of $\Pi$ correspond to well-defined conformal
field theories. In particular, when $\Pi$ is orthogonal to a vector $v\in\Gamma^{4,20}$ of norm
$v^2=-2$ (usually called a root of $\Gamma^{4,20}$), the
corresponding non-linear $\sigma$-model is not well defined \cite{Aspinwall:1995xy}. This subtlety
can be understood by considering the model as an internal CFT in type IIA superstring theory. This theory
is dual to heterotic string theory compactified on ${\mathbb T}^4$. Generically, the corresponding
low energy effective field theory contains an abelian gauge group $U(1)^{24}$. However, when
$\Pi$ is orthogonal to a root $v$, the gauge group is enhanced to a non-abelian gauge group, and
$v$ is interpreted as a root of the corresponding Lie algebra. The additional states in type IIA superstring
theory are interpreted as D-branes becoming massless at this point of the moduli space. This means
that the corresponding perturbative superconformal field theory cannot describe correctly all massless
degrees of freedom of the theory, and hence the non-linear $\sigma$-model is expected to be 
inconsistent \cite{Strominger:1995cz}. Therefore, in the following, we shall exclude the points
in the moduli space where $\Pi$ is orthogonal to a root. It is believed that these are the only 
singular points in the moduli space.
\medskip

Our strategy to characterise the groups $G_\Pi$ is inspired by the proof of  the Mukai theorem given by
Kondo \cite{Kondo}. Let $\Pi\in\RR^{4,20}$ be a $4$-dimensional positive definite space, not orthogonal
to any root (vector of norm $-2$) in $L\equiv \Gamma^{4,20} \subset \RR^{4,20}$, and let
$G\equiv G_\Pi\subset O(\Gamma^{4,20})$ be the subgroup of lattice automorphisms fixing $\Pi$ 
pointwise. We denote by $L^G$ the sublattice of vectors fixed by $G$
\be\label{LuG}
\Gamma^{4,20}\supset L^G:=\{v\in\Gamma^{4,20}|g(v)=v\text{ for all }g\in G\}\ ,
\ee
and by $L_G$ its orthogonal complement
\be\label{LG}
\Gamma^{4,20}\supset L_G:=\{w\in\Gamma^{4,20}|w\cdot v=0\text{ for all }v\in L^G\}\ .
\ee
By definition, the real vector space $L^G\otimes \RR\subset \RR^{4,20}$ will contain $\Pi$,
$\Pi\subset L^G\otimes \RR$, 
and since $\Pi$ has signature $(4,0)$, the orthogonal complement $L_G$ must be negative definite
and have rank at most $20$. Furthermore, every vector in $L_G$ is orthogonal to $\Pi$, so that, by our
assumption on $\Pi$, $L_G$ contains no roots. The proof of the Theorem then proceeds as follows --- 
the relevant details are given in Appendix~\ref{s:latticeproofs}:
\begin{itemize}
\item First (see Appendix~\ref{s:GinOGamma}), we prove that $L_G(-1)$, can be embedded in the
even unimodular lattice $\Gamma^{25,1}$. The action of $G$ on $L_G$ extends to an action on
$\Gamma^{25,1}$, which fixes all vectors orthogonal to $L_G(-1)$ in $\Gamma^{25,1}$. Thus,
$G$ is a subgroup of ${\rm Aut}(\Gamma^{25,1})$.
\item Next (see Appendix~\ref{s:proofConway}), using the fact that $L_G(-1)$ contains no vectors of
norm $2$ and the properties of ${\rm Aut}(\Gamma^{25,1})$, we show that $L_G(-1)$ must be contained
in a positive definite sublattice of $\Gamma^{25,1}$, namely the Leech lattice
$\Lambda$. This is the unique $24$-dimensional even unimodular lattice with no roots and its group
of automorphisms is the Conway group $Co_0$ \cite{Conway}. This group can be obtained by
extending the sporadic finite simple group $Co_1$ \cite{Atlas} of order
\be
|Co_1|=2^{21}\cdot 3^9\cdot 5^4\cdot 7^2\cdot 11\cdot 13\cdot23\sim  4\times 10^{18}\ ,
\ee
by the $\ZZ_2$ symmetry that changes the sign of all vectors in $\Lambda$. Once again, the action of
$G$ on $L_G(-1)$ can be extended to an action on $\Lambda$ which fixes all vectors orthogonal to
$L_G(-1)$. This means, in particular, that the $\ZZ_2$ symmetry that changes the sign of all vectors in
$\Lambda$ is not an element of $G$, since it has no non-trivial fixed vectors. It follows that $G$ is a
subgroup of the finite simple group $Co_1$.
\end{itemize}
\noindent Combining these results thus leads to the natural analogue of the Mukai-Theorem: 
\bigskip

\noindent {\bf Proposition.} {\em The group $G_\Pi$ of 
symmetries of any non-linear $\sigma$-model on K3 is a subgroup of the Conway group 
$Co_1 \subset Co_0={\rm Aut}(\Lambda)$ that fixes pointwise
a sublattice of the Leech lattice $\Lambda$ of rank at least $4$.}
\medskip

In order to give a more precise description of the groups of symmetries $G_\Pi$, one needs a 
detailed classification of the subgroups of $Co_0$ that fix a sublattice of rank $4$ in $\Lambda$. 
The result of this somewhat technical analysis --- the details are explained in  
Appendix \ref{s:proofs} --- is the Theorem stated in the Introduction. 

We should mention that our analysis does not actually prove that $Co_1$ is the smallest possible 
group containing all these symmetry groups. However, some simple considerations on the order 
of the groups $|G_\Pi|$ are sufficient to exclude all maximal subgroups of $Co_1$, except for
$Co_2$. Furthermore, all the cases (i)--(iv) in the Theorem are actually realised by some 
$\N=(4,4)$ model, provided we assume that every four-dimensional subspace $\Pi\subset \RR^{4,20}$, 
not orthogonal to any vector of norm $-2$ in $\Gamma^{4,20}$, leads to a consistent conformal field 
theory (see Appendix~\ref{B.4}). The characterisation of the symmetry groups as given in the Theorem is
therefore optimal.
\bigskip

In the following we shall describe in detail specific examples that realise some of the 
possibilities of the Theorem. In particular,  the ${\mathbb T}^4/\ZZ_2$ orbifold model to be discussed in 
Section~\ref{s:orbifold}   and the $(2)^4$ Gepner model of Section~\ref{s:24} 
are examples of case (i) and have symmetry groups that are not subgroups of 
$\MM_{24}$, while the Gepner model $(1)^6$ that will be studied in Section~\ref{s:16} realises
precisely case (iii). 

\section{The ${\mathbb Z}_2$ Orbifold Model}\label{s:orbifold}

In order to illustrate the general predictions of the Theorem let us consider a few specific
examples. We begin with the ${\mathbb T}^4/\ZZ_2$ orbifold  model where we take 
${\mathbb T}^4$ to be the orthogonal torus at the
self-dual radius, {\it i.e.}\ the four radii take on the self-dual value, and set the $B$-field on the
torus to zero.
We write $\TT^4= \TT^2\times \TT^2$, and label the two $\TT^2$s by $i=1,2$. For
each $\TT^2$ we use complex coordinates, and thus the left-moving bosonic
and fermionic modes are 
\begin{equation}
\alpha^{(i)}_{n} \ , \qquad  \bar{\alpha}^{(i)}_{n}  \ , \qquad \qquad
\psi^{(i)}_{n} \ , \qquad  \bar{\psi}^{(i)}_{n}  \ , \qquad i=1,2 \ ,
\end{equation}
and similarly for the right-movers. (The right-movers are denoted by a tilde.)
The ${\mathbb Z}_2$ orbifold ${\cal I}$ acts as $-1$ an all of these modes. In addition it maps
the momentum ground states ${\cal I} (p_L,p_R) = (-p_L,-p_R)$.

\subsection{The Spectrum in the RR sector}\label{sec:2.1}

We denote the states in the untwisted sector of  $\TT^4=\TT^2\times \TT^2$ by
\be
|p_L,p_R;N,\tilde{N};s;\tilde s\rangle\ ,
\ee
where $p_L=({n}+{w})/\sqrt{2}$,
$p_R=({n}-{w})/\sqrt{2}$ are the left and right momenta at the self dual radius, with
${n}\in{\mathbb Z}^4$ the momentum and
${w}\in{\mathbb Z}^4$ the winding numbers.
Furthermore, $N$ and $\tilde{N}$ denote the left and right oscillator contributions, while
$s,\tilde s$ label the Ramond ground states. More specifically,
\begin{equation}
(s;\tilde  s)=(s_1,s_2;\tilde  s_1,\tilde  s_2) \ ,
\end{equation}
where each $s_i$, $i=1,2$, can take the two values $\pm \tfrac{1}{2}$,
and the zero modes $\bar\psi_0^{(i)}$ and $\psi_0^{(i)}$ map the states with $s_i=\pm \tfrac{1}{2}$ into 
one another; the analogous statement holds for
the $\tilde  s_i$ in the right-moving sector. The RR ground states have charge $(S,\tilde{S})$
with respect to the left- and right-moving $U(1)$-current, where
\begin{equation}
S = \sum_i s_i \ , \qquad \tilde{S} = \sum_i \tilde{s}_i \ .
\end{equation}
To obtain the spectrum of the $\ZZ_2$-orbifold theory we have to project onto states that are even
under the operator $\I$ acting as
\be
\I |p_L,p_R;N,\tilde{N};s;\tilde s\rangle=
(-1)^{|N|+|\tilde{N}|}(-1)^{S+\tilde{S}} \, |-p_L,-p_R;N,\tilde{N};s;\tilde s\rangle\ ,
\ee
where $|N|$ is the total number of oscillators appearing in $N$,
and likewise for $|\tilde{N}|$.
\smallskip

\noindent In the twisted sector the states are labelled by
\be
|i;N,\tilde{N}\rangle \ ,
\ee
where $i=1,\ldots, 16$ distinguishes the $16$ different fixed points of the orbifold, while
$N,\tilde{N}$ denote again the oscillator numbers. The ground states do not carry any charge with 
respect to the left- and right-moving $U(1)$-currents. The orbifold projection in the twisted sector acts as
\be
\I\, |i;N,\tilde{N}\rangle=(-1)^{|N|+|\tilde{N}|} \, |i;N,\tilde{N}\rangle\ .
\ee
It is straightforward to calculate the elliptic genus from this description. In the untwisted sector one finds
\be\label{Z2U}
\phi^{(U)}(\tau,z)=(2y+4+2y^{-1})\prod_{n=1}^\infty\frac{(1+q^ny)^2(1+q^ny^{-1})^2}{(1+q^n)^4}=
8\frac{\vartheta_{2}(\tau,z)^2}{\vartheta_{2}(\tau,0)^2}\ ,
\ee
while the contribution of the twisted sector equals
\be \label{Z2Tw}
\phi^{(T)}(\tau,z)=
8\Bigl(\frac{\vartheta_{4}(\tau,z)^2}{\vartheta_{4}(\tau,0)^2}
+\frac{\vartheta_{3}(\tau,z)^2}{\vartheta_{3}(\tau,0)^2}\Bigr)\ .
\ee
It is easy to see that their sum, $\phi^{(U)}(\tau,z) + \phi^{(T)}(\tau,z)$, reproduces precisely
(\ref{K3eg}).

\subsection{Symmetries and Twining Genera}\label{s:3.2}

The orbifold theory possesses various symmetries that leave the $\N=(4,4)$ superconformal
algebra invariant and that preserve the spectral flow operators; in particular, we have 
\begin{list}{{\rm (\arabic{enumi})}}{\usecounter{enumi}}
\item ${\cal R}$: rotation of the two ${\mathbb T}^2$'s by $90$ and $-90$ degrees, respectively.
\item ${\cal E}$: exchanging the two  ${\mathbb T}^2$'s, together with an inversion acting on the
second ${\mathbb T}^2$, say.
\item $H_a$: half-period translations, that act as $(-1)^{p\cdot a}$,  
$a\in (\ZZ/2\ZZ)^4$, in the untwisted sector, 
and by a permutation on the $16$ twisted sectors.
\item ${\cal T}$: 4-fold T-duality. (Note that the 2-fold T-duality induces a non-trivial automorphism of the
$\N=(4,4)$ superconformal algebra.)
\item ${\cal Q}$: the quantum symmetry that acts as $+1$ on the untwisted, and as $-1$ on the twisted sector.
\end{list}
These symmetries generate the group $G=2^{1+8}\rtimes \ZZ_2^3$. The normal subgroup acting 
trivially on the untwisted sector  RR ground states with $p_L=p_R=0$ 
is the extra special group $2^{1+8}$ of order $2^9$, containing 
$\mathcal{Q}$, the center of $G$. It is generated by the half-shifts $H_a$ that form an abelian $\ZZ_2^4$, 
and by their dual symmetries $G_a=\mathcal{T}H_a\mathcal{T}$ that change the sign of half of the 
twisted sectors. The only non-trivial commutators are $H_aG_bH_aG_b={\cal Q}^{a\cdot b}$, 
where $a,b\in(\ZZ/2\ZZ)^4$. The quotient group $G/2^{1+8}\cong \ZZ_2^3$ is generated by ${\cal R}$, ${\cal E}$ 
and ${\cal T}$. Conjugation by ${\cal R}$ and ${\cal E}$ yields a permutation of the half shifts (and the 
analogous permutations for $G_a$), while conjugation by $\mathcal{T}$ exchanges $G_a$ and $H_a$.
The group $G=2^{1+8}\rtimes \ZZ_2^3$ realises case (i) of the Theorem with
$G'=\ZZ_2^6$ and $G''=\ZZ_2^4 \rtimes \ZZ_2^2$. Here $G''\subset \MM_{24}$ 
is generated by $H_a$, ${\cal E}$ and ${\cal R}$, while 
$G'\subset \ZZ_2^{11}$ is generated by $\Q$ (one $\ZZ_2$ factor), by the composition 
$\T {\cal R} {\cal E}$ (another $\ZZ_2$), and by the elements of the form 
$H_a {\cal T} {\cal R} {\cal E} H_a$ (giving $\ZZ_2^4$).

Given the explicit description of the RR sector from above, it is straightforward to calculate the
corresponding twining genera,
\be
\phi_g(\tau,z)=\Tr_{\rm RR}(g\, q^{L_0-\frac{1}{4}}\bar q^{\bar L_0-\frac{1}{4}}y^{J_0}(-1)^{F+\bar F}) \ ,
\ee
{\it i.e.}\ the elliptic genus with the insertion of the symmetry $g$.\footnote{The condition
that the symmetry preserves the full $\N=(4,4)$ superconformal algebra guarantees that the resulting twining
genus still defines a weak Jacobi form. Symmetries that only preserve the $\N=(2,2)$ superconformal
algebra typically act non-trivially on $J^\pm$,  and then the resulting twining genera do
not have the shift symmetry under $z\mapsto z+\tau$.}
In particular, we can check how
these twining genera compare with the Mathieu twining genera that were worked out in
\cite{Cheng:2010pq,Gaberdiel:2010ch,Gaberdiel:2010ca,Eguchi:2010fg}. Our explicit results are
for example
\begin{eqnarray}
\phi_\R(\tau,z) & = & 4\frac{\vartheta_2(2\tau,2z)}{\vartheta_2(2\tau,0)}
+4\frac{\vartheta_3(2\tau,2z)}{\vartheta_3(2\tau,0)} = \phi_{\rm 2A}(\tau,z) \label{orbR}\\
\phi_{\cal E}(\tau,z)& = & 4\frac{\vartheta_2(2\tau,2z)}{\vartheta_2(2\tau,0)}
+4\frac{\vartheta_3(2\tau,2z)}{\vartheta_3(2\tau,0)} = \phi_{\rm 2A}(\tau,z)\label{orbE} \\
\phi_{H_a}(\tau,z)& = & 8\frac{\vartheta_2(\tau,z)^2}{\vartheta_2(\tau,0)^2}=\phi_{\rm 2A}(\tau,z)\label{orbH} \\
\phi_\T(\tau,z) & = &
-2\vartheta_{4}(2\tau)^4\frac{\vartheta_1(\tau,z)^2}{\eta(\tau)^6}  = \phi_{\rm 2B}(\tau,z)\label{Z2T} \\
\phi_{\T H_{a}}(\tau,z) & = & \phi_{\rm 4A}(\tau,z) \qquad\qquad\qquad \quad
[\hbox{$a^2$ even, say $a= (1100)$} ]\\
\phi_{\R H_{a}}(\tau,z) & = & \phi_{\rm 4B}(\tau,z) \qquad \qquad \qquad \quad
[\hbox{{\it e.g.}\ for  $a=(1111)$}]\\ 
\phi_{\T H_{a}}(\tau,z) & = & \phi_{\rm 4C}(\tau,z)   \qquad \qquad \qquad \quad
[\hbox{{\it e.g.}\ for  $a=(1000)$}] \\
\phi_{\Q}(\tau,z) &= & 
8\Bigl(\frac{\vartheta_{2}(\tau,z)^2}{\vartheta_{2}(\tau,0)^2}
- \frac{\vartheta_{3}(\tau,z)^2}{\vartheta_{3}(\tau,0)^2}
-\frac{\vartheta_{4}(\tau,z)^2}{\vartheta_{4}(\tau,0)^2}\Bigr) \nonumber \\
& = & 2 \phi_{\rm 2A}(\tau,z) - \phi_{\rm 1A} (\tau,z) \qquad \label{Z2Q} \\[4pt]
\phi_{\Q\T H_a G_b}(\tau,z) & = &
-2 \, \frac{\vartheta_3(2\tau)^4\, \vartheta_1(\tau,z)^2}{\eta(\tau)^6} 
- 2 \frac{\vartheta_3(\tau,z)^2}{\vartheta_3(\tau)^2} 
- 2 \frac{\vartheta_4(\tau,z)^2}{\vartheta_4(\tau)^2} \label{Z24p} \\
& = & \tfrac{1}{2}(-\phi_{\rm 1A}+\phi_{\rm 2A}+2\phi_{\rm 4B}) \quad  \;
[\hbox{{\it e.g.}\ for  $a=(1100)$, $b=(0011)$}] \ .  \nonumber
\end{eqnarray}
It is worth pointing out that the generators  $H_a$, ${\cal E}$ and ${\cal R}$ that generate 
$G''\subset \MM_{24}$ lead to twining genera that directly agree with  $\MM_{24}$ twining genera.
On the other hand, the twining genera of the quantum symmetry $\Q$, see (\ref{Z2Q}), and of
the group elements $g =\Q\T H_{a}G_{b}$ for suitable choices of $a$ and $b$, see (\ref{Z24p}), do
not equal   {\em any} Mathieu twining genus 
(but can only be expressed in terms of linear combinations of such twining genera). 
Finally, certain twining genera involving the 4-fold T-duality, namely $\T$ and $\T H_a$ for suitable $a$, 
give rise to  twining genera, whose conjugacy classes (2B, 4A and 4C) do not lie inside $\MM_{23}$. 
This ties in with the fact that T-duality is a non-geometric symmetry.

\subsection{The D-brane Charge Lattice}

As suggested by the proof of the Theorem (and  explained in Section~\ref{symmetries}), 
we should also be able to characterise the symmetry group of the model as the 
symmetries of the D-brane charge lattice that leave the 4-dimensional subspace
$\Pi$ (corresponding to the 4 supercharges) invariant. For the case at hand, the D-brane
charge lattice can be  computed straightforwardly since the primitive branes may be 
taken to be fractional D0, D2 and D4 branes. They can be constructed as explained, for example, in
\cite{Gaberdiel:2000jr}  (see also \cite{Bergman:1999kq}). If we denote the Ishibashi states
in the untwisted and the $i^{\rm th}$ twisted sector by  $| {\rm Dp,U} \rangle\!\rangle$ and
$| {\rm Dp},i \rangle\!\rangle$, respectively, the fractional branes have the structure
\begin{eqnarray}
|\!| {\rm D4},\epsilon, \delta \rangle\!\rangle_{\rm f} & = & \tfrac{1}{2} | {\rm D4,U} \rangle\!\rangle
+ \tfrac{\epsilon}{4} \sum_{i=1}^{16}  \delta_i \, | {\rm D4},i \rangle\!\rangle \\
|\!| {\rm D2(jk)},\epsilon,\delta \rangle\!\rangle_{\rm f} & = & \tfrac{1}{2} | {\rm D2(jk),U} \rangle\!\rangle
+ \tfrac{\epsilon}{2}  \sum_{i \in  P_{jk}} \delta_i\,  | {\rm D2}(jk),i \rangle\!\rangle \\
|\!| {\rm D0},i,\epsilon \rangle\!\rangle_{\rm f} & = & \tfrac{1}{2} | {\rm D0,U} \rangle\!\rangle
+  \epsilon| {\rm D0},i \rangle\!\rangle \ ,
\end{eqnarray}
where $|\!| {\rm D0,i,f} \rangle\!\rangle$ is the fractional D0-brane at the $i^{\rm th}$ fixed point,
while $|\!| {\rm D2(jk),f} \rangle\!\rangle$ denotes the fractional D2-branes oriented along the
(jk) direction, with $P_{jk}$ the set of four fixed points between which the D2-brane is `spanned'.
Furthermore, $\epsilon$ and $\delta_i$ are signs, and the configurations of signs $\delta_i$
arise from Wilson lines, {\it i.e.}\ not all configurations of signs are allowed.
The complete even self-dual charge lattice is spanned by\footnote{The explicit description
of the relevant branes and their intersection matrix is given in the \LaTeX\ source.}
\begin{list}{{\rm (\Roman{enumi})}}{\usecounter{enumi}}
\item 9 fractional D4-branes: one has no Wilson line and all twisted charges $+\tfrac{1}{4}$
($\epsilon=+1$); one has no Wilson line and all twisted charges $-\tfrac{1}{4}$
($\epsilon=-1$); the remaining 7 fractional D-branes have twisted charge $+\tfrac{1}{4}$ 
at the origin, and different choices of Wilson lines.
\item 9 fractional D0-branes: one sits at the origin and has twisted charge $+1$; one sits
at the origin and has twisted charge $-1$; the remaining 7 fractional D0-branes have charge $+1$
and sit at different fixed points.
\item 6 fractional D2-branes, oriented along all $6= {4 \choose 2}$ 2-planes, without any Wilson lines
({\it i.e.}\ $\delta_i=+1$) and positive twisted charge ($\epsilon=+1$).
\end{list}
The non-zero entries of the intersection form are 
\begin{eqnarray}
& & \langle\!\langle {\rm D4}, \! {\rm U}|   {\rm D0}, \! {\rm U} \rangle \!\rangle = 2 \ , \quad \nonumber \\
& & \langle\!\langle {\rm D2(12)}, \!{\rm U}\, |   {\rm D2(34)}, \!{\rm U} \rangle \!\rangle = 
\langle\!\langle {\rm D2(13)},\! {\rm U}\, |   {\rm D2(24)}, \!{\rm U} \rangle \!\rangle = 
\langle\!\langle {\rm D2(14)}, \!{\rm U}\, |   {\rm D2(23)}, \!{\rm U} \rangle \!\rangle = 2 \nonumber \\
&& \langle\!\langle {\rm D}*,i | {\rm D}*,j\rangle\!\rangle = - 2 \delta_{ij} \ .
\end{eqnarray}
It is then straightforward to check that the intersection matrix of the above D-branes has
determinant $1$, {\it i.e.}\ that these $24$ branes generate indeed the full charge lattice. 

The RR charges that transform in the $({\bf 2},{\bf 2})$-representation of \
${\rm SU}(2)_L\times {\rm SU}(2)_R$
and hence span $\Pi$ are precisely those carried by the bulk brane combinations
\be
\bigl( {\rm D0} + {\rm D4} \bigr)\ , \qquad \bigl({\rm D2(12)} +  {\rm D2(34)} \bigr)\ , \qquad
\bigl({\rm D2(13)} +  {\rm D2(24)} \bigr) \ , \qquad
\bigl({\rm D2(14)} +  {\rm D2(23)} \bigr) \ ,
\ee
where D2(ij) denotes the bulk D2-brane oriented along the (ij) direction; this can be
deduced from the analysis of \cite{Brunner:2006tc}, see in particular Appendix~A.2. The
orthogonal complement (with respect to the intersection form) then turns out to be a lattice of rank~20.
Upon changing the sign of its quadratic form, this lattice can be embedded into the Leech lattice 
$\Lambda$. Its orthogonal complement in $\Lambda$ is generated by four vectors 
$y_1,\ldots,y_4\in\Lambda$ with $y_i\cdot y_j=4\delta_{ij}$. According to our general argument, 
the group of symmetries $G$ of the model is isomorphic to the group of automorphisms of 
$\Lambda$ that act trivially on $y_1,\ldots,y_4$. Since $G$ fixes some vectors of norm 
$8$ in $\Lambda$ (for example $y_1+y_2$), it must be a subgroup of $\ZZ_2^{12}\rtimes\MM_{24}$,
see Appendix~\ref{s:proofs}.
A more detailed analysis shows that $G=2^{1+8}\rtimes \ZZ_2^3$, thus matching the results from 
above.

\section{The Gepner Model $(1)^6$}\label{s:16}

Next we consider the $(1)^6$ Gepner model which will turn out to realise case (iii) of the
Theorem. It is constructed by taking six tensor powers of the $\N=2$ minimal
model at $k=1$,  subject to the $\ZZ_3$ orbifold projection generated by
\be \otimes_{i=1}^6\Phi^l_{m_i,s_i;\bar m_i,\bar s_i}\mapsto
e^{\frac{2\pi i}{3}\sum_i m_i}\otimes_{i=1}^6\Phi^l_{m_i,s_i;\bar m_i,\bar s_i}\ .
\ee
For a short review of Gepner models, as well as an outline of our notations and conventions see 
Appendix~\ref{s:Gepner}. 

\subsection{The Spectrum}

The orbifold theory has an untwisted ($\Hh^{(0)}$) and two twisted ($\Hh^{(1)}$
and $\Hh^{(-1)}$) sectors, with spectrum
\be \Hh^{(n)}=\bigotimes_{i=1}^6\Hh_{l_i,m_i+n,s_i}\otimes\bar\Hh_{l_i,m_i-n,\bar s_i}\ .
\ee
Invariance under the orbifold symmetry requires
\be \sum_{i=1}^6 m_i\equiv 0\mod 3\ .
\ee
For $k=1$, we may take $l=0$, with $m\in\ZZ/6\ZZ$ and $s\in\ZZ/4\ZZ$.
The only states that contribute to the elliptic genus are the RR states with
$\bar h=\frac{1}{4}$.
It is easy to see that the condition $\bar h=\tfrac{1}{4}$ together with the $U(1)$-charge integrality
condition (that follows from orbifold invariance) 
is only satisfied if the right-moving ground state is of the form
\be (0,1,1)^{\otimes 6},\quad (0,-1,-1)^{\otimes 6}\ ,\quad(0,1,1)^{\otimes 3}(0,-1,-1)^{\otimes 3}\ .
\ee
In the last case, all the $20$ different permutations of the factors should be considered. Thus,
the RR states for which the right-movers are R ground states are explicitly 
\begin{equation}
\begin{array}{llr}
n=0 : \qquad & \otimes_i^6 \Phi^0_{1,s_i;1,1} & (1\text{ state}) \\[2pt]
                       &\otimes_i^6 \Phi^0_{-1,s_i;-1,-1} & (1\text{ state}) \\[2pt]
            &\bigl(\otimes_i^3 \Phi^0_{1,s_i;1,1}\bigr)\otimes\bigl(\otimes_i^3 \Phi^0_{-1,s_i;-1,-1}\bigr)
            \text{ and permutations} \quad & (20\text{ states}) \\[4pt]
n=1: \qquad  &\otimes_i^6 \Phi^0_{3,s_i;1,1},& (1\text{ state}) \\[2pt]
                        &\otimes_i^6 \Phi^0_{1,s_i;-1,-1} & (1\text{ state}) \\[2pt]
            &\bigl(\otimes_i^3 \Phi^0_{3,s_i;1,1}\bigr)\otimes\bigl(\otimes_i^3 \Phi^0_{1,s_i;-1,-1}\bigr)
            \text{ and permutations} & (20\text{ states})\\[4pt]
n=-1: \qquad  &\otimes_i^6 \Phi^0_{-1,s_i;1,1} & (1\text{ state}) \\[2pt]
                         &\otimes_i^6 \Phi^0_{3,s_i;-1,-1} & (1\text{ state}) \\[2pt]
           &\bigl(\otimes_i^3 \Phi^0_{-1,s_i;1,1}\bigr)\otimes\bigl(\otimes_i^3 \Phi^0_{3,s_i;-1,-1}\bigr)
           \text{ and permutations} & (20\text{ states})
\end{array}
\end{equation}
where $s_i=1,3$. 
In each case, states with $s_i=1$ and $s_i=3$, $i=1,\ldots,6$, are mapped into one another 
under the action of the  $\N=2$ algebra of the $i^{\rm th}$ model. The contribution  to the elliptic genus of
the $\N=2$ representation containing $\otimes_i^6\Phi^0_{m_i,s_i;\bar m_i,\bar s_i}$ is
\be
\prod_i^6I^0_{m_i}(\tau,z)I^0_{\bar m_i}(\bar\tau,0) \ ,
\ee
where $I^l_m(\tau,z)$ are the $\N=2$ `characters' that are defined in (\ref{CharacterN2rep}). Since
\be\label{signs}
I^0_{-1}(\bar\tau,0)=-1\ ,\qquad I^0_{1}(\bar\tau,0)=1\ ,\qquad I^0_{3}(\bar\tau,0)=0\ ,
\ee
we obtain
\be \phi(\tau,z)=2 \sum_{\substack{m\in\ZZ/6\ZZ\\m\text{ odd}}}I^0_m(\tau,z)^6
-20\sum_{\substack{m\in\ZZ/6\ZZ\\m\text{ odd}}} I^0_m(\tau,z)^3I^0_{m+2}(\tau,z)^3\ ,
\ee
which reproduces indeed (\ref{K3eg}).

\subsection{Symmetries and Twining Genera}

Let us first describe the symmetries that  preserve the $\N=2$ superconformal symmetry. For 
$i\in\{1,\ldots, 6\}$, we  denote by $e_i$ the $\ZZ_3$ symmetry that acts as
\be
e_i(\otimes_{j=1}^6 \Phi^0_{m_j,s_j;\bar m_j,\bar s_j})
=e^{\frac{2\pi i m_i}{3}}(\otimes_{j=1}^6 \Phi^0_{m_j,s_j;\bar m_j,\bar s_j})\ .
\ee
These symmetries generate a group $\ZZ_3^5$ (because the product of all of them is the orbifold symmetry,
under which all states are invariant by construction). There are also the right-moving analogs
$\bar e_i$, where the phase depends on $\bar m_i$ instead of $m_i$. In addition we have the
quantum symmetry $\Q$ which acts by multiplication by $e^{\frac{2\pi i n}{3}}$ on the $n^{\rm th}$ twisted
sector $\Hh^{(n)}$. This gives an additional $\ZZ_3$. Note that in the $n^{\rm th}$ twisted sector
$\bar m_i=m_i-2n$, and thus
\be
\bar e_i=\Q e_i\ ,
\ee
implying that the transformations $\Q,e_i,\bar e_i$ form a group $\ZZ_3^6$.
Finally, there are permutations of the six factors of $(1)^6$, where
we define the action of $\pi\in S_6$ by 
\be\label{Jpm}
 \pi \bigl(\otimes_{i=1}^6 \Phi^l_{m_i+n,0;m_i-n,0}\bigr)
 ={\rm sgn}(\pi)^{n} \bigl(\otimes_{i=1}^6
 \Phi^l_{m_{\pi(i)}+n,0;m_{\pi(i)}-n,0}\bigr)\ ,
\ee
with additional signs for $s_i,\bar s_i\neq 0$, given by the usual bosonic/fermionic
statistics.

Thus, the group of symmetries preserving the $\N=2$ superconformal symmetry is $\ZZ_3\times \ZZ_3^5.S_6$.
The subgroup preserving the $\N=4$ superconformal algebra is then generated by the transformations
that leave the currents $J^+$ and $J^-$ of the $\N=4$ superconformal algebra
\be
(\Phi^0_{-2,2;0,0})^{\otimes 6} \in \Hh^{(- 1)} \ ,\qquad (\Phi^0_{2,2;0,0})^{\otimes 6} \in \Hh^{(+1)}
\ee
invariant. This group has the structure $\ZZ_3^4\rtimes A_6$ and is generated by\footnote{We have
also checked that these symmetries preserve the spectral flow operators.}
\begin{list}{{\rm (\arabic{enumi})}}{\usecounter{enumi}}
\item The phase transformations
\be
\prod_{i=1}^6 e_i^{n_i} \qquad \text{with}\quad \sum_{i=1}^6 n_i\equiv 0\mod 3 \ ,
\ee
where the constraint assures invariance of the currents. Because of the orbifold invariance
relation $\prod_{i=1}^6 e_i=1$, they generate the normal subgroup $\ZZ_3^4$.
\item The even permutations, since the odd permutations act on the states (\ref{Jpm}) with a minus sign.
The even permutations form the alternating group $A_6$. 
\end{list}
The resulting group is the semidirect product $\ZZ_3^4\rtimes A_6$, where
we have the obvious action $\pi(\prod_{i=1}^6 e_i^{n_i})=\prod_{i=1}^6 e_{\pi(i)}^{n_i}$ of $A_6$ on 
the generators of $\ZZ_3^4$. The $(1)^6$ Gepner model is therefore an example of case (iii) of
the Theorem.
\medskip

With this description it is now straightforward to calculate the corresponding twining genera, and 
for the convenience of the reader we have collected them in Table~\ref{Tab:TwiningGenera16}. 
For example, we have for the phase transformation $e_1e_2e_3$ 
\begin{align}
\phi_{e_1e_2e_3}(\tau,z)&=2\sum_{m=-1,1,3}I_{m}^6-(2+9e^{2\pi i/3}+9e^{-2\pi i/3})
\sum_{m=-1,1,3}I_m^3I_{m+2}^3\\
&= 2\sum_{m=-1,1,3}I_{m}^6+7\sum_{m=-1,1,3}I_m^3I_{m+2}^3
=-\tfrac{1}{2}\phi_{{\rm 1A}}+\tfrac{3}{2}\phi_{{\rm 3A}} \equiv \hat\phi_{\rm 3a}\ ,
\end{align}
which corresponds to the fourth line of the table. The other twining genera that do 

\begin{table}[htb]
\begin{center}
\begin{tabular}{|c|l||c|c|c|c|}\hline
\multicolumn{2}{|c||}{\textbf{symmetry}} & \multicolumn{4}{|c|}{\textbf{properties}}\\\hline\hline
perm. & conditions on phases & \textbf{$N$} & \textbf{\#} & $\text{Tr}_{\mathbf{24}}$ & $\phi_g$ \\\hline\hline
$ijklmn$ & \parbox{7cm}{\vspace{0.2cm}$n_i=n_j=n_k=n_l=m_m=n_n=0$\vspace{0.1cm}} 
&  $1$  & $1$ & $24$& $\phi_{\rm 1A}$ \\\hline
\parbox{1.95cm}{\vspace{0.1cm}$(ij)(kl)mn$\vspace{0.1cm}} 
& \parbox{7cm}{\vspace{0.2cm}$n_i+n_j=n_k+n_l=n_m=n_n=0$ \vspace{0.1cm}} 
& \parbox{0.25cm}{$2$} & \parbox{0.65cm}{$405$} & \parbox{0.25cm}{$8$}& $\phi_{\rm 2A}$ \\\hline
\parbox{1.95cm}{\vspace{0.1cm}$\phantom{()}ijklmn$\\ $\phantom{)}(ijk)lmn$\\ $(ijk)(lmn)$\vspace{0.1cm}} 
& \parbox{7cm}{\vspace{0.2cm}$n_i=2\,,n_j=1\,,n_k+n_l+n_m+n_m=0$ \\
 $n_i+n_j+n_k=0\,,n_l+n_m+n_n=0$ \\ 
$n_i+n_j+n_k=0\,,n_l+n_m+n_n=0$\vspace{0.1cm}} & \parbox{0.25cm}{$3$} 
 & $2220$ & \parbox{0.25cm}{$6$}& $\phi_{\rm 3A}$ \\\hline
$ijklmn$ & \parbox{7cm}{\vspace{0.2cm}$n_i=n_j=n_k=1\,,n_l=n_m=n_n=0$\vspace{0.1cm}} 
 & \parbox{0.25cm}{$3$} & $20$ &$-3$& $\hat\phi_{\rm 3a}$ \\\hline
$(ijkl)(mn)$ & \parbox{7cm}{\vspace{0.2cm}$n_i+n_j+n_k+n_l+n_m+n_n=0$\vspace{0.1cm}} 
 & \parbox{0.25cm}{$4$} & $7290$ & \parbox{0.25cm}{$4$}& $\phi_{\rm 4B}$ \\\hline
$(ijklm)n$ & \parbox{7cm}{\vspace{0.2cm}$n_i+n_j+n_k+n_l+n_m+n_n=0$\vspace{0.1cm}} 
 & \parbox{0.25cm}{$5$} & $11664$ & \parbox{0.25cm}{$4$}& $\phi_{\rm 5A}$ \\\hline
\parbox{2cm}{\vspace{0.1cm}$(ij)(kl)mn$\\ $(ij)(kl)mn$\\ $(ij)(kl)mn$\\$(ij)(kl)mn$\\ $(ij)(kl)mn$\vspace{0.1cm}} 
 & \parbox{7cm}{\vspace{0.2cm}$n_i+n_j=n_k+n_l=0\,,n_m=1\,,n_n=2$ \\$n_i+n_j=n_k+n_l=1\,,n_m=1\,,n_n=0$ \\ 
 $n_i+n_j=n_k+n_l=2\,,n_m=2\,,n_n=0$\\$n_i+n_j=2\,,n_k+n_l=0\,,n_m+n_n=1$\\
 $n_i+n_j=1\,,n_k+n_l=0\,,n_m+n_n=2$\vspace{0.1cm}} & \parbox{0.25cm}{$6$} & $1620$ 
 & \parbox{0.25cm}{$2$}& $\phi_{\rm 6A}$ \\\hline
\parbox{2cm}{\vspace{0.1cm}$(ij)(kl)mn$\vspace{0.1cm}} 
 & \parbox{7cm}{\vspace{0.2cm}$n_i+n_j=2\,,n_k+n_l=0\,,n_m=1\,,n_n=0$\vspace{0.1cm}} 
 & \parbox{0.25cm}{$6$} & \parbox{0.8cm}{$1620$} & \parbox{0.25cm}{$5$}
 & $\hat\phi_{\rm 6a}$ \\\hline
 \parbox{2cm}{\vspace{0.1cm}$(ijk)(lmn)$\vspace{0.1cm}} 
 & \parbox{7cm}{\vspace{0.2cm}$n_i+n_j+n_k=2\,,n_l+n_m+n_n=1$\vspace{0.1cm}} 
 & \parbox{0.25cm}{$9$} & \parbox{0.8cm}{$3240$} & \parbox{0.25cm}{$3$}& $\hat\phi_{\rm 9a}$ \\\hline
\parbox{1.6cm}{\vspace{0.1cm}$(ijk)lmn$\\ $(ijk)lmn$\vspace{0.1cm}} 
 & \parbox{7cm}{\vspace{0.2cm}$n_i+n_j+n_k=2\,,n_l+n_m+n_n=1$ \\ 
 $n_i+n_j+n_k=1\,,n_l+n_m+n_n=2$\vspace{0.1cm}} & \parbox{0.25cm}{$9$} & $1080$ 
 & \parbox{0.25cm}{$3$}& $\hat\phi_{\rm 9b}$ \\\hline

\end{tabular}
\caption{Twining genera of the $(1)^6$ model. The symmetry generators have been labelled by the 
 structure of the permutations $\{i,j,k,l,m,n\}$ of the minimal models and the individual phase shifts $e_i^{n_i}e_j^{n_j}e_k^{n_k}e_l^{n_l}e_m^{n_m}e_n^{n_n}$. The multiplicity (labelled by $\#$) 
 is the number of `independent' generators within each class of symmetries which are not identified 
 through the action of the orbifold. The order of each generator is denoted by $N$ and 
 $\text{Tr}_{24}$ gives the trace over the $24$-dimensional representation. Finally, 
 the twining genera $\hat\phi_{\rm 3a}$,  $\hat\phi_{\rm 6a}$ and $\hat\phi_{\rm 9ab}$ 
 are not Mathieu twining genera and  are defined in the main body of the text.} 
\label{Tab:TwiningGenera16}
\end{center}
\end{table}
 
\noindent not directly agree with $\MM_{24}$ twining genera --- since we are in case (iii) of the 
Theorem, there is no reason to expect any $\MM_{24}$ twining genera --- are
\begin{align}
\hat\phi_{\rm 6a}(\tau,z)&=\tfrac{1}{2}\left(\phi_{\rm 2A}(\tau,z)+\phi_{\rm 6A}(\tau,z)\right) \\
\hat\phi_{\rm 9a}(\tau,z)& =\tfrac{1}{2}\left(\phi_{\rm 3A}(\tau,z)+\phi_{\rm 3B}(\tau,z)\right) \\
\hat\phi_{\rm 9b}(\tau,z)&=\tfrac{1}{4}\left[\phi_{0,1}(\tau,z)+6\bigl(2\psi^{(3)}(\tau)+3\psi^{(9)}(\tau)
+3E_2^{(9)}(\tau)\bigr)\phi_{-2,1}(\tau,z)\right] \ ,
\end{align}
where $\phi_{0,1}$ and $\phi_{-2,1}$ are the standard Jacobi forms of ${\rm SL}(2,\mathbb{Z})$
of index one and weight 0 and $-2$, respectively, and $\{\psi^{(3)},\psi^{(9)},E_2^{(9)}\}$ form a basis of 
weight $2$ modular forms of $\Gamma_0(9)$ (for more information and explicit definitions see 
\cite{Gaberdiel:2010ca}).

\subsection{The D-brane Charge Lattice}

As before, we can also determine the symmetry group of this model from an analysis of
the D-brane charge lattice. The construction of Gepner model D-branes is standard,
and is briefly sketched in Appendix~\ref{s:Gepner}. The
tensor product A-type branes with $L_i=S_i=0$ and 
$M_1,\ldots,M_6\in \ZZ_6$ generate a charge lattice of rank $22$ with signature $(2,20)$. This is 
as expected, since the A-type tensor product branes only couple to the $22$ RR ground states 
in the untwisted sector.

The other charges are carried by B-type permutation branes of the type described in 
(\ref{16Bper}). As we vary
$M_1,\ldots,M_6,\hat M\in\ZZ_6$, the intersection form of the B-type branes
gives a matrix of rank $10$ with signature $(2,8)$. Again, this is what we expect
since  these permutation branes couple to the $2$ RR ground states in the twisted sectors
$n=\pm 1$, and to the $8$ RR ground states in the untwisted sector with 
$m_1=-m_2$, $m_3=-m_4$, and $m_5=-m_6$.

In order to obtain the full charge lattice we have to combine these two constructions; for example, 
a set of $22$ A-type D-branes with $L_i=0=S_i$, $i=1,\ldots,6$ and suitable values for $M_1,\ldots,M_6$,
and two B-type D-branes with $L_i=M_i=0$, $S_i=0$ and $\hat M=\pm2$ generate the full 
unimodular lattice $\Gamma^{4,20}$ (see the \LaTeX\ source code for details).

Next, we denote that four RR ground states in the $({\bf 2},{\bf 2})$ representation of 
${\rm SU}(2)_L\times {\rm SU}(2)_R$ according to their $J_0^3,\tilde{J}_0^3$ charges as 
$\Phi_{1,\bar 1}$, $\Phi_{1,-\bar 1}$, $\Phi_{-1,\bar 1}$, 
and $\Phi_{-1,-\bar 1}$,  Let us consider the sublattice $(\Gamma^{4,20})^\perp$ of D-branes that are 
neutral under these four states. The
$22$ A-type branes generate the sublattice of D-branes that are neutral with respect to
$\Phi_{1,-\bar 1}$ and  $\Phi_{-1,\bar 1}$, while their charge with 
respect to $\Phi_{1,\bar 1}$ and $\Phi_{-1,-\bar 1}$ is given by
\begin{align} Q_{\Phi_{1,\bar 1}}(\|0,M_i,0\rrangle_{\rm A})&= \frac{1}{3}e^{\frac{\pi i}{3}\sum_i M_i}\ ,\\
Q_{\Phi_{-1,-\bar 1}}(\|0,M_i,0\rrangle_{\rm A})&= \frac{1}{3}e^{-\frac{\pi i}{3}\sum_i M_i}\ .
\end{align}
With the redefinition
\be
\Phi^a=3 (\Phi_{1,\bar 1}+\Phi_{-1,-\bar 1})\ ,\qquad \Phi^b=2 \sqrt{3} i(\Phi_{1,\bar 1}-\Phi_{-1,-\bar 1})\ ,
\ee
the states $\Phi^a,\Phi^b$ correspond to elements of the dual lattice
$(\Gamma^{4,20})^*\cong \Gamma^{4,20}$ so that the sublattice generated by A-type
D-branes and orthogonal to these elements has maximal rank,
\be
\mathrm{rk}(\Gamma^{4,20})^\perp=20\ .
\ee
The discriminant group of this $20$-dimensional lattice is $\ZZ_9\times\ZZ_3^2$, and its discriminant form is the 
same as the one of the S-lattice $2^93^6$ (see Appendix~\ref{s:proofs}), whose quadratic form is
\be\label{Q2923} Q_{2^{9}3^{6}}=\begin{pmatrix}
 \; \; 4 & \; \; 1 & \; \; 1 & -2 \\
\; \;  1 & \; \; 4 & \; \; 1 & -2 \\
 \; \; 1 & \; \;  1 & \; \; 4 & \; \;  1 \\
 -2 & -2 & \; \;  1 & \; \;  4
\end{pmatrix}\ .
\ee
By the general lattice gluing procedure \cite{Conway}, it follows that the S-lattice $2^93^6$ 
and $(\Gamma^{4,20})^\perp(-1)$ are orthogonal sublattices of a positive definite even 
unimodular lattice of rank $24$. This procedure also provides explicitly the quadratic form for this 
unimodular lattice. With the help of some computer algorithm we have  shown that the 
resulting lattice has no vectors of norm $2$, thus proving that it is indeed the Leech lattice 
$\Lambda$. All  sublattices of $\Lambda$ with quadratic form \eqref{Q2923} are related to the 
S-lattice $2^93^6$ by some Leech lattice automorphism, and the pointwise stabiliser of each of them 
is the group $G=\ZZ_3^4\rtimes A_6$ \cite{Curtis}, thus matching the results from the previous 
subsection. 

\section{The $(2)^4$ Model}\label{s:24}

Our last example is the `quartic' model $(2)^4$ which is constructed by taking a 
$\mathbb{Z}_4$ orbifold of four tensor powers of the $\N=2$ minimal model of level 
$k=2$. Our notation and conventions are the same as in the previous section and are again 
summarised in Appendix~\ref{s:Gepner}. 

\subsection{The Spectrum}

For the $\mathbb{Z}_4$ orbifold we have in addition to the  
untwisted sector $\mathcal{H}^{(0)}$ three twisted sectors ($n=1,2,3$), with spectrum 
\be \Hh^{(n)}=\bigotimes_{i=1}^4\Hh_{l_i,m_i+n,s_i}\otimes\bar\Hh_{l_i,m_i-n,\bar s_i}\ ,
\ee
where invariance under the $\mathbb{Z}_4$-orbifold  enforces $\sum_{i=1}^4 m_i\equiv 0\mod 4$. 
At $k=2$ there are six $\N=2$ R-sector representations, which we may label 
by $(l=0, m=\pm 1, \pm 3)$ and $(l=1,m=0,2)$. For the elliptic genus we are again only interested in those 
RR states for which the right-moving states are ground states with $\bar h=\frac{1}{4}$. The relevant 
coset representations are 
(compare also~\cite{Brunner:2006tc,Brunner:2009mn})
\begin{align}
&\hspace{1.2cm}(0,1,1)^{\otimes 4}\,,\hspace{1.5cm} (1,2,1)^{\otimes 4}\,,\hspace{1.5cm} (0,-1,-1)^{\otimes 4}\,,\\
&(0,1,1)^{\otimes 2}\otimes (0,-1,-1)^{\otimes 2}\,,\hspace{1cm} (0,1,1)\otimes (0,-1,-1)\otimes (1,2,1)^{\otimes 2}\,,
\end{align}
where in the second line all $6$ and $12$ permutations are included, respectively. Explicitly,
the RR states that can contribute to the elliptic genus are thus of the form 
\begin{equation}
\begin{array}{llll}
n=0: & n=1: & n=2: & n=3: \\
\otimes_i^4 \Phi^0_{1,s_i;1,1} \qquad &  \otimes_i^4 \Phi^0_{3,s_i;1,1} \qquad & 
\otimes_i^4 \Phi^0_{-3,s_i;1,1}  \qquad &  \otimes_i^4 \Phi^0_{-1,s_i;1,1}  \\
\otimes_i^4 \Phi^1_{2,s_i;2,1} \qquad &  \otimes_i^4 \Phi^1_{0,s_i;2,1} \qquad &
\otimes_i^4 \Phi^1_{-2,s_i;2,1} \qquad &  \otimes_i^4 \Phi^1_{0,s_i;2,1}  \\
\otimes_i^4 \Phi^0_{-1,s_i;-1,-1}  \qquad &  \otimes_i^4 \Phi^0_{1,s_i;-1,-1} \qquad &
\otimes_i^4 \Phi^0_{3,s_i;-1,-1}  \qquad & \otimes_i^4 \Phi^0_{-3,s_i;-1,-1} \ .
\end{array}
\end{equation}
In addition, there are the states (again written in the order $n=0$, $n=1$, $n=2$, and $n=3$)
\begin{equation}
\begin{array}{ll}
\bigl(\otimes_i^2 \Phi^0_{1,s_i;1,1}\bigr)\otimes\bigl(\otimes_i^2 \Phi^0_{-1,s_i;-1,-1}\bigr) \qquad &
\bigl(\otimes_i^2 \Phi^0_{3,s_i;1,1}\bigr)\otimes\bigl(\otimes_i^2 \Phi^0_{1,s_i;-1,-1}\bigr) \\
 \bigl(\otimes_i^2 \Phi^0_{-3,s_i;1,1}\bigr)\otimes\bigl(\otimes_i^2 \Phi^0_{3,s_i;-1,-1}\bigr) \qquad &
\bigl(\otimes_i^2 \Phi^0_{-1,s_i;1,1}\bigr)\otimes\bigl(\otimes_i^2 \Phi^0_{-3,s_i;-1,-1}\bigr)\ ,
\end{array}
\end{equation}
where in each case there are 6 different permutations,  as well as the states
\begin{equation}
\begin{array}{ll}
\bigl(\otimes_i^2 \Phi^1_{2,s_i;2,1}\bigr)\otimes \Phi^0_{-1,s_i;-1,-1}\otimes \Phi^0_{1,s_i;1,1} \quad &
\bigl(\otimes_i^2 \Phi^1_{0,s_i;2,1}\bigr)\otimes \Phi^0_{1,s_i;-1,-1}\otimes \Phi^0_{3,s_i;1,1} \\
 \bigl(\otimes_i^2 \Phi^1_{-2,s_i;2,1}\bigr)\otimes \Phi^0_{3,s_i;-1,-1}\otimes \Phi^0_{-3,s_i;1,1} \quad &
 \bigl(\otimes_i^2 \Phi^1_{0,s_i;2,1}\bigr)\otimes \Phi^0_{-3,s_i;-1,-1}\otimes \Phi^0_{-1,s_i;1,1} \ ,
 \end{array}
 \end{equation}
where now there are 12 different permutations each. Again, the states with $s_i=1,3$ are mapped into
one another under the action of the  $\N=2$ algebra of the $i^{\rm th}$ model. Since
\be
I^0_{\pm 1}(\bar\tau,0)=\pm 1 \ , \quad  \quad I^0_{\pm 3}(\bar\tau,0)=0 \ , \quad
I^1_{0}(\bar\tau,0)=0 \ , \quad I^1_{\pm 2}(\bar\tau,0) =\pm 1  
\ee
the total contribution to the elliptic genus is then
\be
\phi = \sum_{\substack{m\in\ZZ/8\ZZ\\m\text{ odd}}}
\Bigl[ 2 (I^0_m)^4 + 6 (I^0_m)^2 (I^0_{m-2})^2 
- 12 I^0_m I^0_{m+2}  (I^1_{m+3})^2 +(I^1_{m+1})^4\Bigr] \ ,
\ee
which agrees  indeed with (\ref{K3eg}).

\subsection{Symmetries and Twining Genera}

In the quartic $(2)^4$ model the currents $J^\pm$ of the left- and right moving $\mathcal{N}=4$ 
superconformal algebra  are given by
\begin{equation}
J^\pm= (0,\pm2,2)^{\otimes4}\otimes \overline{(0,0,0)}^{\otimes 4} \ , \qquad \hbox{and} \qquad
\bar{J}^\pm=(0,0,0)^{\otimes4}\otimes\overline{(0,\pm2,2)}^{\otimes 4}\ .\label{N4current24}
\end{equation}
The symmetries that leave these currents invariant are
\begin{list}{{\rm (\arabic{enumi})}}{\usecounter{enumi}}
\item Phase shifts, which are generated by 
\be
\prod_{i=1}^{4} e_i^{a_i} \, \bigl( \Q^2 (-1)^{F_s} \bigr)^{\frac{A}{2}} \ , \qquad
A = \sum_{i=1}^{4} a_i \qquad \hbox{with} \qquad
A \equiv 0 \,\,\, \hbox{mod} \, 2 \ .
\ee
Here each $e_{i}^{a_i}$ acts as 
\be
e_{i}^{a_i}:\,\Phi^\ell_{m_i,s_i;\bar{m}_i,\bar{s}_i}\longmapsto 
e^{\frac{2\pi i m_i a_i}{4}}\Phi^\ell_{m_i,s_i;\bar{m}_i,\bar{s}_i}\  ,
\ee
$\Q$ is the quantum symmetry of the Gepner orbifold (that acts as a phase 
$e^{\frac{\pi i n}{2}}$ on the states of the $n^{\text{th}}$ twisted sector), and 
$(-1)^{F_s}$ is the left-moving spacetime fermion number operator that acts
as $+1$ ($-1$) on the left-moving NS (R) sector. (The inclusion of $(-1)^{F_s}$
is required in order to preserve the spectral flow operators.)
Taking into account the overall $\mathbb{Z}_4$ invariance coming from the Gepner orbifold,
these phase shifts generate the group $\mathbb{Z}_4^2\times \ZZ_2$.
\item The permutations 
\be\label{Jpm4}
 \pi \bigl(\otimes_{i=1}^4 \Phi^l_{m_i+n,0;m_i-n,0}\bigr)={\rm sgn}(\pi)^{n} 
 \Q^{ {\rm sgn}(\pi)-1}  \bigl(\otimes_{i=1}^4
 \Phi^l_{m_{\pi(i)}+n,0;m_{\pi(i)}-n,0}\bigr)\ .
\ee
They generate the symmetric group $S_4$.
\end{list}
These symmetries generate the group $(\ZZ_2\times \ZZ_4^2)\rtimes S_4$, thus realising 
case (i) of the Theorem with $G'=\ZZ_2^3$ and $G''=\ZZ_2^2.S_4$.  Here 
$G'\subset \ZZ_2^{11}$ is generated by the phases $e_1^2e_2^2$,    
$e_2^2e_3^2$   and   $e_3^2\Q^2(-1)^{F_s}$, while 
$G''\subset \MM_{24}$ is generated by the permutations, giving the 
$S_4$ factor, as well as by the phases $e_1e_2\Q^2(-1)^{F_s}$ and $e_2^3e_3^3\Q^2(-1)^{F_s}$, giving the 
$\ZZ_2^2$ factor.\footnote{Note that these last two phases are order $2$ only after taking the quotient by 
$G'$; in fact, their squares are non-trivial elements in $G'$, so that as elements of the whole group 
$G$ they are order 4.}

We can also calculate the associated twining genera,
and our results are collected in Table~\ref{Tab:TwiningGenera}. Again, we see that some of the generators in
$G'$ lead to the  twining genus 
\be
\hat{\phi}_{4a} =-\frac{1}{2}\left(\phi_{\rm 1A}-\phi_{\rm 2A}-2\phi_{\rm 4B}\right) \ ,
\ee
that does not coincide with any twining genus of $\MM_{24}$. 

\begin{table}[htbp]
\begin{center}
\begin{tabular}{|c|l|c|c||c|c|c|c|}\hline
\multicolumn{4}{|c||}{\textbf{symmetry}} & \multicolumn{4}{|c|}{\textbf{properties}}\\\hline\hline
perm. & conditions on phases & $n_Q$ & $n_F$ & \textbf{$N$} & \textbf{\#} & 
$\text{Tr}_{\mathbf{24}}$ & $\phi_g$  \\\hline\hline
$ijkl$ & \parbox{4.7cm}{\vspace{0.2cm}$n_i=n_j=n_k=n_l=0$\vspace{0.2cm}} &  
$0$ & $0$ & $1$  & $1$ & $24$& $\phi_{\rm 1A}$ \\\hline
\parbox{1.3cm}{\vspace{0.1cm}$\phantom{()}ijkl$\\ $(ij)(kl)$\\$\phantom{(}(ij)kl$\\
$\phantom{(}(ij)kl$\\$\phantom{()}ijkl$\vspace{0.1cm}} & 
\parbox{7cm}{\vspace{0.1cm}$n_i=n_j=2\,,n_k+n_l=0$\\$n_i+n_j=n_k+n_l\in\{0,2\}$\\
$n_i+n_j=n_k+n_l=0\,,n_{k,l}\in\{0,2\}$\\$n_i+n_j=0\,,n_k+n_l=2\,,n_{k,l}\in\{0,2\}$\\
$n_i+n_j+n_k+n_l=2\,,n_{i,j,k,l}\in\{0,2\}$\vspace{0.1cm}} & \parbox{0.25cm}{$0$\\$0$\\$2$\\$0$\\$2$} 
& \parbox{0.25cm}{$0$\\$0$\\$0$\\$1$\\$1$} & \parbox{0.25cm}{$2$} & $79$ & \parbox{0.25cm}{$8$}& 
$\phi_{\rm 2A}$ \\\hline 
$(ij)(kl)$ & \parbox{4.7cm}{\vspace{0.1cm}$n_i+n_j=n_k+n_l\in\{1,3\}$\vspace{0.1cm}} 
&  $2$ & $1$ & $2$  & $24$ & $0$& $\phi_{\rm 2B}$ \\\hline
$(ijk)l$ & \parbox{5cm}{\vspace{0.1cm}$n_i+n_j+n_k+n_l=0$\vspace{0.1cm}} 
& $0$ & $0$ & $3$  & $128$ & $6$& $\phi_{\rm 3A}$ \\\hline
$ijkl$ & \parbox{5.5cm}{\vspace{0.1cm}$n_i=n_j\in\{1,3\}\,,n_k=n_l=0$\vspace{0.1cm}} 
& $2$ & $1$ & $4$  & $6$ & $-4$& $\hat{\phi}_{4a}$ \\\hline
\parbox{1.3cm}{\vspace{0.1cm}$\phantom{()}ijkl$\\ $\phantom{(}(ijkl)$\\ 
$(ij)(kl)$\\$\phantom{(}(ij)kl$\\$\phantom{(}(ij)kl$\\$\phantom{(}(ij)kl$\\$\phantom{(}(ij)kl$\\
$\phantom{()}ijkl$\vspace{0.1cm}} & \parbox{7cm}{\vspace{0.1cm}$n_i=1\,,n_j+n_k+n_l=3$\\
$n_i+n_j+n_k+n_l=0$\\$n_i=n_j=1\,,n_k+n_l=2$\\$n_i+n_j+n_k+n_l=2\,,n_{i,j}\in\{1,3\}$\\
$n_i+n_j+n_k+n_l=2\,,n_{k,l}\in\{1,3\}$\\$n_i+n_j=n_k+n_l=0\,,n_{i,j}\in\{1,3\}$\\
$n_i+n_j=n_k+n_l=0\,,n_{k,l}\in\{1,3\}$\\$n_i=1\,,n_j+n_k+n_l=1$\vspace{0.1cm}} & 
\parbox{0.25cm}{$0$\\$2$\\$0$\\$0$\\$0$\\$2$\\$2$\\$2$} & 
\parbox{0.25cm}{$0$\\$0$\\$0$\\$1$\\$1$\\$0$\\$0$\\$1$} & \parbox{0.25cm}{$4$} & $306$ 
& \parbox{0.25cm}{$4$}& $\phi_{\rm 4B}$ \\\hline 
$(ijk)l$ & \parbox{5cm}{\vspace{0.1cm}$n_i+n_j+n_k+n_l=2$\vspace{0.1cm}} & $2$ & $1$ & $6$  & $128$ & 
$2$& $\phi_{\rm 6A}$ \\\hline
$(ij)kl$ & \parbox{5cm}{\vspace{0.1cm}$n_i+n_j=1\,,n_k+n_l=3$\vspace{0.1cm}} & 
$2$ & $0$ & $8$  & $96$ & $2$& $\phi_{\rm 8A}$ \\\hline
\end{tabular}
\caption{The twining genera of the $(2)^4$ model. 
Here the symmetries have been labelled by the structure of the 
permutations of $\{i,j,k,l\}$, the phase shifts of the individual symmetries 
$e_i^{n_i}e_j^{n_j}e_k^{n_k}e_l^{n_l}$, the power of the operator 
$Q^{n_Q}$ and the spacetime fermion number $(-1)^{n_F F_s}$. The remaining part of the 
notation is the same as for Table~1.}
\label{Tab:TwiningGenera}
\end{center}
\end{table}

\subsection{The D-brane Charge Lattice}

The derivation of a set of D-branes generating the lattice of RR charges is analogous to the 
construction for the $(1)^6$ model. The A-type tensor product branes (see Appendix~\ref{s:Dbranes})
are now only charged under the $21$ RR ground states in the 
untwisted sector. The remaining charges can be accounted for in terms of B-type permutation
branes. Taking $21$ A-type branes with $L_i=S_i=0$ and suitable combinations for $M_i$,
as well as $3$ B-type permutation branes with $L_i=M_i=S_i=0$ and suitable values
of $\hat{M}$ leads indeed to the full charge lattice $\Gamma^{4,20}$,
{\it i.e.}\ the resulting intersection form has determinant one.\footnote{See the \LaTeX\ source code 
for details; there we also give further details about the sublattice $(\Gamma^{4,20})^\perp$ and its 
embedding into the Leech lattice.}

The sublattice $(\Gamma^{4,20})^\perp$ of D-branes that are neutral with respect to the RR ground
states in the $({\bf 2},{\bf 2})$ representation of ${\rm SU}(2)_L\times {\rm SU}(2)_R$ has maximal rank $20$. 
Upon changing  the sign of its quadratic form, it can be embedded into the Leech lattice $\Lambda$, and 
its orthogonal complement $\Lambda^G$ has quadratic form
\be \begin{pmatrix}
 6 & \;\; 2 & \;\; 0 & \;\; 0 \\
 2 & \;\; 4 & -2 & \;\; 4 \\
 0 & -2 & \;\; 6 & -4 \\
 0 & \;\; 4 & -4 & \;\; 8
\end{pmatrix}\ .
\ee Since $\Lambda^G$ contains a vector of norm $8$, its point-wise stabiliser must be a subgroup of 
$\ZZ_2^{12}\rtimes\MM_{24}$ (see Appendix~\ref{s:proofs}). 
More precisely, we have shown that the stabiliser turns out to be isomorphic to 
$G=(\ZZ_2\times\ZZ_4^2)\rtimes S_4$, which is the group of symmetries we have found in the 
previous  subsection.

\section{Conclusions}

In this paper, we have shown that the symmetries of  a
non-linear $\sigma$-model on K3 that preserve the $\N=(4,4)$-superconformal
algebra as well as the spectral flow operators,  form a subgroup of the Conway group $Co_1$. This
provides a stringy analogue of the Mukai theorem in algebraic geometry that shows that the 
symplectic automorphisms of any K3 form a subgroup of the Mathieu group $\MM_{23}$. The specific
subgroups that can actually arise in our case are spelled out in the Theorem stated in the Introduction.

Our result is somewhat unexpected in view of the recent observation of \cite{EOT}, 
relating the elliptic genus of K3 to the Mathieu group $\MM_{24}$. In particular, it follows
from our Theorem (as well as the explicit examples) that the symmetries of a given K3 model are
not, in general, subgroups of $\MM_{24}$.\footnote{Apparently this was also independently noted
by the authors of \cite{EOT}; we thank Yuji Tachikawa for discussions about this point.} As a consequence,
their twining genera do not, in general, agree with those appearing in the context of  Mathieu moonshine 
\cite{Cheng:2010pq,Gaberdiel:2010ch,Gaberdiel:2010ca,Eguchi:2010fg}, and we have seen
explicit examples of this. In particular, this therefore means that the naive idea that  $\MM_{24}$ arises
as the `union' of all symmetries from different points in moduli space needs to be refined. 

At least on the face of it, our Theorem seems to suggest that the elliptic genus of K3 could exhibit 
some sort of moonshine based on $Co_1$ or $Co_2$, but we have seen no evidence of this since
the dimensions of their representations do not match the coefficients of the elliptic genus. It is intriguing
that a connection between $Co_1$ and the BKM algebras arising in $\TT^6$-compactifications of
the heterotic string has  recently been observed in \cite{Govindarajan:2011mp}; given that the heterotic
string on $\TT^6$ is dual to type IIA on $K3\times \TT^2$ this could be related to our findings.

Our analysis also provides useful tools for the general understanding of non-linear 
$\sigma$-models on K3. For example, our Theorem suggests the existence of models 
with some large symmetry groups, and gives precise predictions for their lattice of 
D-brane charges. These predictions were nicely verified in the three examples we considered. 
In particular, we showed  that the $(1)^6$ Gepner model realises case (iii) of the Theorem. 
Some preliminary investigations suggest that the groups described in case (iv) might be realised in 
terms of $\TT^4/\ZZ_3$ torus orbifolds for different choices of metric and B-field, while it is more difficult 
to guess which model realises case (ii). For case (i), the Theorem predicts the existence of a 
model with symmetry group $\ZZ_2^8\rtimes\MM_{20}$, that might correspond to a 
certain $\TT^4/\ZZ_2$ orbifold. 

On more general grounds, the sublattice of D-branes that are neutral under the RR ground 
states  in the $({\bf 2},{\bf 2})$ representation of SU$(2)_L\times {\rm SU}(2)_R$ is,  
in a certain sense, the stringy analogue of the Picard lattice in algebraic geometry. Since the groups 
of symmetries have a genuine action on this sublattice, it would be interesting to understand 
for which models this lattice has maximal rank $20$. For example, one can show that for a Gepner model of 
type $(k_1)\cdots (k_r)$, a necessary condition for this  to happen is that the greatest common divisor 
$\gcd(k_1+2,k_2+2,\ldots)$ of their shifted levels is $3$, $4$ or $6$.

\section*{Acknowledgements} 
We thank Michael Douglas, Martin Fluder, Daniel Persson and Yuji Tachikawa for useful discussions and 
correspondences. The research of MRG is supported by the Swiss National Science Foundation.

\appendix

\section{Notation and Mathematical Background}\label{s:notation}

\subsection{Group Theory}

Let us give a brief summary of our conventions regarding finite groups. 
\begin{center}
\begin{tabular}{rl}
$A\times B$ & \parbox{12.5cm}{The direct product of the groups $A$ and $B$.}\\[10pt]
$N\rtimes H$ & \parbox{12.5cm}{The semidirect product of $H$ acting on the normal subgroup $N$.}\\[10pt]
\parbox{0.9cm}{$N.Q$ \\ \\ } & \parbox{12.5cm}{A group $G$ having $N$ as a normal subgroup such 
that $G/N\cong Q$. (This notation includes the direct and semidirect product as special subcases.)}\\[20pt]
$\ZZ_n$ & \parbox{12.5cm}{Cyclic group of order $n$.}\\[10pt]
\parbox{0.9cm}{$p^{1+2n}_+$ \\ \\ \\ \\ } & \parbox{12.5cm}{Extra-special group of 
order $p^{1+2n}$ (we will always omit the plus in our notation). For a prime $p$ and 
positive integer $n$, $p^{1+2n}$ is the extension of $\ZZ_p^{2n}$ by a central element 
$z$ of order $p$. It is generated by $2n$ elements $x_1,\ldots,x_n,y_1,\ldots,y_n$ of order $p$, 
with $x_ix_j=x_jx_i$, $y_iy_j=y_jy_i$ and 
$\prod x_i^{r_i}\prod y_j^{s_j}=z^{r\cdot s}\prod y_j^{s_j}\prod x_i^{r_i}$.}\\[40pt]
\end{tabular}
\end{center}
\begin{center}
\begin{tabular}{rl}
$S_n$ & \parbox{12.5cm}{Group of permutations of $n$ elements (symmetric group).}\\[10pt]
$A_n$ & \parbox{12.5cm}{Group of even permutations of $n$ elements (alternating group).}\\[10pt]
\parbox{0.7cm}{$\MM_{24}$ \\ \\ \\} & \parbox{12.5cm}{The largest Mathieu group, a sporadic simple 
group of order $2^{10}\cdot 3^3\cdot 5\cdot 7\cdot11\cdot23=244823040$. It can be described as a 
group of permutations of $24$ elements. More precisely, it is the subgroup of $S_{24}$ that preserves the 
binary Golay code. This group has $26$ conjugacy classes.}\\[35pt]
\parbox{0.7cm}{$\MM_{23}$ \\} & \parbox{12.5cm}{The subgroup of $\MM_{24}$ that, in its representation 
as a permutation of $24$ elements, fixes one element.}\\[20pt]
\parbox{0.7cm}{$Co_0$ \\ \\ } & \parbox{12.5cm}{The Conway group is the automorphism group 
${\rm Aut}(\Lambda)$ of the Leech lattice. It is an extension of $Co_1$ by the central $\ZZ_2$ that flips the 
sign of all vectors in $\Lambda$.}\\[20pt]
\parbox{0.7cm}{$Co_1$ \\ } & \parbox{12.5cm}{A sporadic simple group of order \\
$2^{21}\cdot 3^9\cdot 5^4\cdot7^2\cdot11\cdot13\cdot23=4157776806543360000$.}\\[15pt]
\parbox{1.65cm}{$\ZZ_2^{12}\rtimes \MM_{24}$\\ \\ \\ \\ \\} & \parbox{12.5cm}{A maximal subgroup of 
${\rm Aut}(\Lambda)$. It fixes a set of $24$ mutually orthogonal vectors $x_1,\ldots, x_{24}\in\Lambda$ of 
norm $8$ ($x_i\cdot x_j=8\delta_{ij}$), up to signs. Each element of the normal subgroup $\ZZ_2^{12}$ 
changes the sign of $n$ of these $24$ vectors, with $n=0,8,12,16,24$. This description establishes 
a one to one correspondence of $\ZZ_2^{12}$ with the binary Golay code \cite{Conway}. The subgroup 
$\MM_{24}$ acts by permutations of $x_1,\ldots,x_{24}$.}
\end{tabular}
\end{center}
\medskip

\subsection{Lattices}\label{s:lattices}

For any lattice $L$, we denote by $L^*$ its dual lattice and by $L(n)$, $n\in\RR$, the lattice obtained
 from $L$ by multiplying the quadratic form by $n$. If $L$ is integral, then $L\subseteq L^*$, and the 
 finite abelian group $A_L=L^*/L$ is called the \emph{discriminant group}. We denote by 
 $l(L)$ the minimal number of generators of $A_L$ (notice that $l(L)\le \rk L$, with $\rk L$ the
 rank of the lattice ). 
 
Let $A_L$ be the discriminant group of an integral lattice $L$, and $q_L$ the 
associated \emph{discriminant quadratic form}, {\it  i.e.}\ the form 
\be 
q_L:A_L\to \QQ/2\ZZ
\ee
induced by the quadratic form on $L$. More generally, we denote by $A_q$ a finite abelian 
group with a quadratic form $q:A_q\to \QQ/2\ZZ$. The quadratic form $q$ determines a bilinear 
form on $A_q$ which takes values in $\QQ/\ZZ$; we denote it by $a\cdot b$, where $a,b\in A_q$. 

If $L$ is a sublattice of a unimodular lattice $\Gamma$ and $L^\perp$ is its orthogonal complement in 
$\Gamma$, then there is an isomorphism between the discriminant groups
\be\label{A.2}
\gamma:A_L\stackrel{\cong}{\to} A_{L^\perp}\ ,
\ee 
that flips the sign of the quadratic form
\be\label{A.3} 
q_L=-q_{L^\perp}\circ\gamma\ .
\ee
More precisely, $\bar x\cong \bar y$, with $\bar x\in A_L$ and $\bar y\in A_{L^\perp}$, if and only if
 $x+y\in\Gamma$ for any choice $x\in L^*$, $y\in(L^\perp)^*$ of representatives of $\bar x,\bar y$. 
 Conversely, given two even lattices $L_1$, $L_2$ with isomorphic discriminant groups 
 $\gamma:A_{L_1}\stackrel{\cong}{\to} A_{L_2}$ and  opposite discriminant quadratic forms
 $q_{L_1}=-q_{L_2}\circ\gamma$, one can construct an even unimodular lattice 
 $\Gamma$ by `gluing' $L_1$ and $L_2$, {\it i.e.}\
\be 
\Gamma=\{x\oplus y\in L_1^*\oplus L_2^*\mid \bar x\cong \bar y\}\ ,
\ee 
where $\bar x,\bar y$ are the images in $A_{L_1}$, $A_{L_2}$ of $x\in L_1^*$ and $y\in L_2^*$,
respectively.

A sublattice $L'$ of a lattice $L$ is called \emph{primitive} if $L/L'$ is a free group;
 in other words, $L'=(L'\otimes \QQ)\cap L$. Correspondingly, a primitive embedding of a lattice 
 $L'$ in $L$ is an embedding such that the image is primitive.
 
 \subsection{The Leech Lattice and the Golay Code}
 
The Leech lattice can be defined in terms of the binary Golay code $\C_{24}$, a 
$12$-dimensional subspace of the vector space $\FF_2^{24}$, where $\FF_2=\{0,1\}$ is the field 
with two elements. To each vector $f=(f_1,\ldots,f_{24})$ of $\C_{24}$ ({\it i.e.}\ to each codeword) 
we associate a subset $X_f$ of $\Omega=\{1,\ldots,24\}$, corresponding to the 
non-zero coordinates of $f$, {\it i.e.}\ $X_f=\{i\in\Omega\mid f_i\neq 0\}$. This 
collection of $2^{12}=4096$ subsets of $\Omega$ (called $\C$-sets) includes the empty set, 
$\Omega$ itself, 759 $\C$-sets with $8$ elements (special octads), $2576$ with $12$ elements and 
759 with $16$ elements (the complements in $\Omega$ of the special octads). Furthermore, 
for any choice of $5$ distinct elements in $\Omega$ there is a unique special octad containing them. 
These properties are sufficient to determine the collection of $\C$-sets, and thus the Golay code, 
up to permutations of the objects in $\Omega$. The Mathieu group $\MM_{24}$ can be defined 
as the group of automorphisms of the Golay code, or, equivalently, as the subgroup of $S_{24}$ 
stabilising the collection of $\C$-sets.
\smallskip

An explicit description of the Leech lattice $\Lambda\subset\RR^{24}$ can be given as follows: the
vector $v=\frac{1}{\sqrt{8}}(v_1,\ldots,v_{24})$ is an element of the Leech lattice provided that 
\begin{itemize}
\item[--] the $v_i$, $i=1,\ldots,24$, are all integers of the same parity; 
\item[--] $\sum_{i=1}^{24}v_i\equiv 0$ or $4\mod 8$ according to $v_i\equiv 0$ or $1\mod 2$, respectively; and 
\item[--] for each $\nu\in\{0,1,2,3\}$, the set $\{i\in \Omega\mid v_i\equiv \nu\mod 4\}$ is a $\C$-set.
\end{itemize}

\section{Proof of the Theorem}\label{s:latticeproofs}

In this appendix, we give the remaining details of the arguments of Section~\ref{symmetries}, leading 
to the proof of the Theorem.

\subsection{$G$ as a Subgroup of $O(\Gamma^{25,1})$}\label{s:GinOGamma}

For each lattice $L$ and group $G\subset {\rm Aut}(L)$, we define as in (\ref{LuG})  the 
$G$-invariant lattice $L^G$ as $L^G=\{x\in L\mid g(x)=x,\forall g\in G\}$. Furthermore, $L_G$ is its 
orthogonal complement in $L$, $L_G=\{x\in L\mid x\cdot y=0,\forall y\in L^G\}$. We take 
$L=\Gamma^{4,20}\subset\RR^{4,20}$, and consider the case where  $G$ is the subgroup of 
${\rm Aut}(\Gamma^{4,20})\subset O(4,20,\RR)$, fixing a positive-definite four plane 
$\Pi\subset\RR^{4,20}$.

\begin{proposition}\label{p:lemma} For any choice of a positive 4-plane  $\Pi$, $L_G$ is a negative 
definite lattice of rank  $\rk L_G\le 20$. $G$ acts trivially on $A_{L_G}$,  and 
$l(L_G)\le 24-\rk(L_G)$.
\end{proposition}
\begin{proof} The first part is obvious. By definition $G$ acts trivially on $L^G$, and hence 
its induced action on  $A_{L^G}$ is also trivial. Since $L^G$ and $L_G$ are orthogonal 
primitive sublattices of the self-dual lattice $L$, it follows that for each $y\in (L_G)^*$ there exists 
a vector $v=x+y\in L$ with $x\in (L^G)^*$. For all $w\in L^G$ and any lattice automorphism $g\in G$, 
$w\cdot g(v)= g(w)\cdot g(v) = w\cdot v$, so that $g(v)-v\in(L^G)^\perp=L_G$. Since 
$g$ is linear and fixes $x\in(L^G)^*$, we have $g(v)-v=g(x+y)-(x+y)=g(y)-y$. It follows that 
$g(y)\equiv y\mod L_G$, so that $G$ acts trivially on $A_{L_G}=(L_G)^*/L_G$. Finally, 
$l(L_G)=l(L^G)\le \rk (L^G)$,  and the last statement follows.
\end{proof}

Up to isomorphism, there is a unique even unimodular lattice $\Gamma^{25,1}$ of signature $(25,1)$. 
The lattice $\Gamma^{25,1}$ can be defined as the (additive) subgroup of $\RR^{25,1}$ with elements 
$(x_0,\ldots,x_{24};x_{25})$ such that 
\be 
x_0+\ldots +x_{24}-x_{25}\in 2\ZZ\ ,
\ee 
where either $x_i\in\ZZ$ for all $i$, or $x_i\in\ZZ+\frac{1}{2}$ for all $i$. 
In the rest of this subsection, we will prove that $L_G(-1)$ can be embedded into $\Gamma^{25,1}$.
\bigskip

Recall that, for any prime $p$, a $p$-group is a group whose order is a power of $p$. A Sylow 
$p$-subgroup of a group $G$ is a maximal $p$-subgroup, {\it i.e.}\ a $p$-subgroup of $G$ 
which is not a proper subgroup of any other $p$-subgroup. For abelian groups, there is a unique 
Sylow $p$-subgroup for each prime $p$, the subgroup of elements whose order is a power of $p$.
For more general finite groups, for each given $p$ the Sylow $p$-subgroups are all isomorphic and 
related by conjugation.

Let $A_q$ be a finite abelian group with quadratic form $q:A_q\to \QQ/2\ZZ$. For any prime $p$, 
let $A_{q_p}$ be the Sylow $p$-subgroup of $A_q$, and $q_p$ the restriction of $q$ to $A_{q_p}$. 
Note that if $a\in A_{q_p}$, then for any $b\in A_q$ 
\be 
p^n\, (a\cdot b)\equiv 0\mod \ZZ\ ,
\ee 
where $p^n$ is the order of $a$, and  $a\cdot b$ is the bilinear form induced by $q$. 
In particular, if $b\in A_{q_{p'}}$, with $p'\neq p$, this implies 
$a\cdot b\equiv 0\mod \ZZ$. Thus, we have an orthogonal decomposition of the quadratic form 
$q=\oplus_pq_p$, where each $q_p$ is a quadratic form on an abelian $p$-group.

It follows from Theorem~1.12.2 of \cite{Nikulin} that  an 
even lattice $L$ of signature $(t^+,t^-)$ and discriminant group $A_q$ can be primitively embedded 
into some even unimodular lattice of signature $(d^+,d^-)$, provided that
\begin{subequations}\label{embedcond}\begin{align}
&d^+-d^-\equiv 0\mod 8\label{embedcond1}\\
&d^--t^-\ge 0,\quad d^+-t^+\ge 0,\quad d^++d^--t^--t^+\ge l(A_q),\label{embedcond2}\\
&d^-+d^+-t^--t^+> l(A_{q_p})\quad  \text{ for all odd primes $p$},\label{embedcond3}\\
&d^-+d^+-t^--t^+> l(A_{q_2})\quad\text{or}\quad q_2=q_\theta^{(2)}(2)\oplus q_2'\ 
\text{ for some }q_2' \ ,\label{embedcond4} 
\end{align}\end{subequations}
where $q_\theta^{(2)}(2)$ is the discriminant quadratic form of the $\mathfrak{su}(2)$ root lattice 
$A_1\cong\ZZ(2)$.

\begin{proposition}\label{th:embed} The lattice $L_G(-1)$ can be primitively embedded into 
$\Gamma^{25,1}$. The action of $G$ can be extended to an action on $\Gamma^{25,1}$ such that the 
$G$-invariant sublattice $(\Gamma^{25,1})^G$ is the orthogonal complement of $L_G(-1)$ in 
$\Gamma^{25,1}$,  and such that $(\Gamma^{25,1})^G$ contains an element of norm $2$.
\end{proposition}
\begin{proof} Let us prove that we can embed $L_G(-1)\oplus A_1$ into $\Gamma^{25,1}$, where 
$A_1$ denotes the root lattice of the $\mathfrak{su}(2)$ Lie algebra. 
Condition \eqref{embedcond1} obviously holds. By proposition \ref{p:lemma}, we have 
$\rk(L_G(-1)\oplus A_1)+l(L_G(-1)\oplus A_1)\le 26$, so that also  \eqref{embedcond2} is satisfied. 
Let $A_q$ be the discriminant group of $L_G(-1)\oplus A_1$, with discriminant quadratic form $q$. 
Let us consider the decomposition $q=\oplus_{p}q_p$ where, for each prime $p$, $q_p$ is the 
restriction of $q$ to the $p$-Sylow subgroup $A_{q_p}$. Since the discriminant group of 
$A_1$ is $\ZZ/2\ZZ$,  we have $l(A_{q_p})\le l(L_G)<26-\rk(L_G\oplus A_1)$ for all odd $p$, and 
\eqref{embedcond3} holds. Finally, it is clear that $q_2=q_\theta^{(2)}(2)\oplus q_2'$, where 
$q_\theta^{(2)}(2)$ is the discriminant form of $A_1$ and $q_2'$ is the restriction of $q_{L_G}$ to the 
$2$-Sylow subgroup. Thus, it follows from \eqref{embedcond} that $L_G(-1)\oplus A_1$ can be 
primitively embedded into $\Gamma^{25,1}$. 

Since $G$ acts trivially on $A_{L_G}$, the action of $G$ on $L_G(-1)$ can be extended to an action on 
$\Gamma^{25,1}$ which acts trivially on the orthogonal complement of $L_G(-1)$ in $\Gamma^{25,1}$. 
Thus, $(\Gamma^{25,1})^G\cong (L_G(-1))^\perp$, and $A_1$ is a sublattice of $(\Gamma^{25,1})^G$, 
so that $(\Gamma^{25,1})^G$ contains a vector of norm $2$.
\end{proof}

\subsection{$G$ as a Subgroup of $Co_0$}\label{s:proofConway}

The automorphism group ${\rm Aut}(\Gamma^{25,1})$ of $\Gamma^{25,1}$ is 
generated by the sign flip $x_{i}\mapsto -x_{i}$, together with the subgroup of autochronous transformations 
${\rm Aut}^+(\Gamma^{25,1})$ which stabilise the cone of positive time vectors in $\RR^{25,1}$ 
(see chapter 27 of \cite{Conway}). The group ${\rm Aut}^+(\Gamma^{25,1})$ contains a normal 
subgroup $W$ (the Weyl group), generated by the reflections $R_r$ with respect to the hyperplanes
$r^\perp$
\be 
R_r(x)=x-(x\cdot r)r\ ,\qquad x\in \Gamma^{25,1}\ ,
\ee 
where $r$ is any root in  $r\in \Gamma^{25,1}$, {\it i.e.}\ satisfies $r\cdot r=2$. 
The complement in $\RR^{25,1}$ of the union $\bigcup_{r\cdot r=2}r^\perp$ of the corresponding hyperplanes 
has infinitely many connected components, and the closure of each component is called a Weyl chamber.

The group of autochronous transformations is the semidirect product $W\rtimes Co_{\infty}$, 
where $Co_\infty$ is the automorphism group of the Dynkin diagram of $W$. 
The groups $W$ and $Co_{\infty}$ can be described more explicitly upon choosing a set of generators 
for $W$, {\it i.e.}\ a set of fundamental roots. One convenient choice is given by the 
set of Leech roots, {\it i.e.}\ by the vectors $r\in \Gamma^{25,1}$ with
\be 
r\cdot r=2\qquad r\cdot w=-1\ ,
\ee 
where $w$ is the null (Weyl) vector
\be 
w=(0,1,2,3,\ldots,23,24;70)\in \Gamma^{25,1}\ .
\ee 
The sublattice $\Gamma^{25,1}\cap w^\perp$ is degenerate, while its quotient 
$(\Gamma^{25,1}\cap w^\perp)/w$ is the Leech lattice 
$\Lambda$ (see chapter 26 of \cite{Conway}), the unique 
positive even  unimodular lattice of rank $24$ containing no roots.
The automorphism group $Co_\infty$ of the Dynkin diagram of $W$ contains 
${\rm Aut}(\Lambda)=Co_0$, the automorphism group of the Leech lattice, as the subgroup 
which fixes a given reference Leech root $\bar r$. The group $Co_{\infty}$ is generated by 
$Co_0$, together with the translations of the Leech roots by vectors in $\Lambda$.
\medskip

Let us consider the embedding of $L_G(-1)\oplus A_1$ into $\Gamma^{25,1}$. Clearly, the sign flip does
not fix any linear sublattice of $\Gamma^{25,1}$, so that $G\subseteq {\rm Aut}^+(\Gamma^{25,1})$.  Note 
that  the lattice $(\Gamma^{25,1})^G$ always contains a vector in the interior of some Weyl chamber. For if this
was not true, then $(\Gamma^{25,1})^G$ would be contained in one of the hyperplanes orthogonal 
to some root $r$, and thus $r\in (\Gamma^{25,1})_G=L_G(-1)$. But this would contradict our 
assumption that $L_G$ contains no vector of norm $-2$.
Since $W$ acts transitively on the Weyl chambers, 
we can choose our embedding of $L_G(-1)$  into $\Gamma^{25,1}$ such that $(\Gamma^{25,1})^G$ 
contains a vector in the interior $K$ of the fundamental Weyl chamber containing $w$. Since 
$K\cap t(K)=\emptyset$, for all non-trivial $t\in W$, it follows that $G$ must be contained in 
$Co_\infty$. Since 
$w$ is fixed by $Co_\infty$, we have $(\Gamma^{25,1})_G\subset (\Gamma^{25,1}\cap w^\perp)$, and 
the projection $(\Gamma^{25,1}\cap w^\perp)\to (\Gamma^{25,1}\cap w^\perp)/w\cong\Lambda$ 
induces an embedding of $(\Gamma^{25,1})_G$ into the Leech lattice $\Lambda$
\be 
L_G(-1)\cong(\Gamma^{25,1})_G\subset (\Gamma^{25,1}\cap w^\perp)\to
 (\Gamma^{25,1}\cap w^\perp)/w\cong \Lambda\ .
\ee  As in proposition \ref{th:embed}, the action of $G$ on $(\Gamma^{25,1})_G$ can be 
extended to an action on $\Lambda$, such that 
$L_G(-1)\cong(\Gamma^{25,1})_G\cong\Lambda_G$ is the orthogonal complement of the 
sublattice $\Lambda^G\subset \Lambda$ fixed by $G$. 
We conclude that  $G$ is a subgroup of ${\rm Aut}(\Lambda)\cong Co_0$ fixing a 
sublattice $\Lambda^G$ of rank at least $4$, 
thus proving the Proposition in Section~\ref{symmetries}.

\subsection{The Proof of the Theorem}\label{s:proofs}

The stabilisers of sublattices of the Leech lattice have been classified \cite{Curtis,Atlas}. 
We will use this classification to prove now the Theorem stated in the Introduction.
The action of ${\rm Aut}(\Lambda)\cong Co_0$ is well defined on the classes of the 
quotient $\Lambda/2\Lambda$, because 
$2\Lambda$ is stable under lattice automorphisms. In particular, if $v\in \Lambda$ is fixed by the action of 
$G$, then $G$ must be in the stabiliser of the class of $\Lambda/2\Lambda$ containing $v$. Note that
opposite vectors $x,  -x\in\Lambda$ are contained in the same class in this quotient. More 
generally, if we define the \emph{short vectors} $x\in\Lambda$ to be the vectors of norm $x^2\le 8$, 
then for each non-trivial class in $\Lambda/2\Lambda$, one of the following mutually exclusive 
alternatives holds \cite{Conway}:
\begin{enumerate}
\item the class contains exactly one pair of short vectors $\pm x\in\Lambda$ of type 2 ($x^2=4$);
\item the class contains exactly one pair of short vectors $\pm x\in\Lambda$ of type 3 ($x^2=6$);
\item the class contains exactly $24$ pairs of short vectors $\pm x_1,\ldots,\pm x_{24}\in\Lambda$ of type $4$, that are mutually orthogonal ($x_i\cdot x_j=8\delta_{ij}$).
\end{enumerate}
\noindent Thus, each primitive vector $v\in \Lambda^G$ is congruent modulo $2\Lambda$ to some short 
vector $v_s\in\Lambda$, $v_s^2\le 8$. We can now distinguish the following cases:
\smallskip

\noindent {\bf Case 1:} 
Suppose there is a vector $v\in\Lambda^G$ which is congruent modulo $2\Lambda$ to a short 
vector of norm $8$. Then the class of $v$ in $\Lambda/2\Lambda$ contains $24$ mutually 
orthogonal pairs $\pm x_1,\ldots,\pm x_{24}\in\Lambda$ of norm $8$. The subgroup of $Co_0$ 
stabilising such a class is the semidirect product $N= \ZZ_2^{12}\rtimes \MM_{24}$.
Here, $\MM_{24}$ acts by permutations of the $24$ pairs $\pm x_1,\ldots,\pm x_{24}$, 
while each element $\epsilon_f\in \ZZ_2^{12}$ is associated to a codeword 
$f\equiv (f_1,\ldots,f_{24})\in (\ZZ/2\ZZ)^{24}$ in the binary Golay code (\cite{Conway}, chapter 11) 
and acts by
\be\label{Golayaction} 
\epsilon_f(x_i)= (-1)^{f_i}x_i\ ,\qquad i=1,\ldots,24 \ ,
\ee
on the vectors of the class. Thus, $G$ is a subgroup of $\ZZ_2^{12}\rtimes \MM_{24}$ that fixes a 
subspace of dimension at least $4$. This realises case (i) of the Theorem. In particular, 
$G''\subset\MM_{24}$ can be any subgroup with at least four orbits $\Omega_1,\ldots, \Omega_4$ 
when acting on $\{1,\ldots,24\}$, and $G'$ is generated by the $\epsilon_f\in\ZZ_2^{12}$ such that 
$f_i=0$ for all $i\in\Omega_1\sqcup\ldots \sqcup\Omega_4$.
\bigskip

\noindent {\bf Case 2:} If Case 1 does not apply, then 
each primitive vector $v\in\Lambda^G$ is congruent modulo $2\Lambda$ to a pair of short vectors 
$\pm v_s\in\Lambda$ with $v_s^2\le 6$. Let us assume that, for each  $v\in \Lambda^G$, the corresponding 
short vector $v_s$ is also contained in $\Lambda^G$. Since $\Lambda^G$ is primitive, 
$(v-v_s)/2\in\Lambda$ is also contained in $\Lambda^G$. The sublattices 
$S\subset \Lambda$ containing only short vectors of norm $4$ and $6$ and such that each primitive 
vector is congruent to a short one mod $2S$ are called $\calS$-lattices, and they have been completely 
classified \cite{Curtis}. In particular, up to automorphisms, there are only three $\calS$-lattices of rank at 
least $4$ \cite{Atlas}:
$$\begin{array}{cccc}
    S \qquad & \rk S \qquad & {\rm Stab}(S) \qquad & {\rm Aut}(S)\\
    2^93^6 \qquad & 4 \qquad & \ZZ_3^4\rtimes A_6\qquad  & \ZZ_2\times(S_3\times S_3).\ZZ_2\\
    2^{5}3^{10} \qquad & 4 \qquad & 5^{1+2}.\ZZ_4 \qquad & \ZZ_2\times S_5\\
    2^{27}3^{36} \qquad  & 6 \qquad & 3^{1+4}.\ZZ_2 \qquad & \ZZ_2\times U_4(2).\ZZ_2 \ .
\end{array}
$$
Since $\Lambda^G$ is a $\calS$-lattice of rank at least $4$, $G$ must be one of the groups 
${\rm Stab}(S)$ from above, corresponding to the cases (ii), (iii) and (iv) (with $G''$ trivial) of 
the Theorem.
\medskip

\noindent {\bf Case 3:} 
The last case arises if each primitive vector $v\in\Lambda^G$ is congruent modulo $2\Lambda$ to a pair 
of short vectors $\pm v_s\in\Lambda$ with $v_s^2\le 6$, but some of these short vectors are not 
contained in $\Lambda^G$. In this case, we define a finite chain of sublattices of $\Lambda$
\be\label{latticechain}
\Lambda^G=S_0\subset S_1\subset\ldots \subset S_N=S\ ,
\ee
where each $S_{i+1}$ is defined in terms of $S_i$ as follows \cite{Allcock}:
\begin{enumerate}
\item[(1)] If $S_i$ is contained in an $\calS$-lattice of the same rank or if $S_i$ contains a vector 
congruent modulo $2\Lambda$ to a short vector of norm $8$, then we set $S=S_i$ and the 
procedure stops.
\item[(2)] Otherwise, if $S_i$ contains a vector $v\in 2\Lambda$ with $v/2\notin S_i$,
then $S_{i+1}$ is obtained by adjoining $v/2$ to $S_i$. There are only a finite number of 
vectors to check, namely one representative for each class in $S_i/2S_i$. Note 
that $S_i\otimes\QQ=S_{i+1}\otimes\QQ$, so that 
$\rk S_i=\rk S_{i+1}$ and ${\rm Stab}(S_i)={\rm Stab}(S_{i+1})$.
\end{enumerate}
If neither of these cases applies, then there is a nonempty set of short vectors 
$v_s\notin S_i\otimes\QQ$, with $0< v_s^2\le 6$, congruent modulo $2\Lambda$ to some 
$v\in S_i$. Then $S_{i+1}$ is defined as follows:
 \begin{enumerate}\setcounter{enumii}{2}
\item[(3)] If one of the short vectors is such that $v_s\cdot w\neq 0$ for some $w\in S_{i}$, then 
$S_{i+1}$ is obtained by adjoining $v_s$ to $S_i$. Any element $g\in {\rm Stab}(S_i)$ must preserve 
the class of $v_s$ in $\Lambda/2\Lambda$, so that $g(v_s)\in\{\pm v_s\}$, and also the product 
$v_s\cdot w\neq 0$. Thus, the only possibility is $g(v_s)=v_s$, so that 
${\rm Stab}(S_{i+1})={\rm Stab}(S_i)$. Furthermore, we have a strong inequality 
$\rk S_{i+1}>\rk S_i$.
\item[(4)] If all the short vectors $v_s$ are orthogonal to $S_i$, we choose one of them and define 
$S_{i+1}$ as the $\ZZ$-linear span of $S_i$ and $v_s$. The only non-trivial action of an element 
$g\in {\rm Stab}(S_i)$ on $S_{i+1}$ is $g(v_s)=-v_s$, so that $|{\rm Stab}(S_i):{\rm Stab}(S_{i+1})|\le 2$.
Furthermore, $\rk S_{i+1}>\rk S_i$.
 \end{enumerate}
The stabiliser ${\rm Stab}(S)$ of the lattice $S$ at the end of the chain must be a subgroup of 
$\ZZ_2^{12}\rtimes\MM_{24}$ or one of the stabilisers of the $\calS$-lattices above. However, 
if the case (4) of the above procedure occurs for some intermediate $S_i$, ${\rm Stab}(S)$ might 
be just a normal subgroup of $G$. Our analysis is greatly simplified by the following result, 
which is an immediate consequence of Lemma 4.8 of \cite{Allcock}.
\begin{proposition}
If $\rk S_i>3$ and none of the cases (1), (2) and (3) applies, then $S_i$ is a sublattice of the 
$\calS$-lattice $2^{27}3^{36}$.
\end{proposition} 
It is easy to see that if $S_i$ is a sublattice of an $\calS$-lattice also $S_{i+1}$ is. Since the starting point
$S_0=\Lambda^G$ of the chain has already rank greater than $3$, this proposition implies that if 
$\Lambda^G$ is not a sublattice of $2^{27}3^{36}$, then case (4) above never occurs. Thus, in this case, 
$G\cong {\rm Stab}(S)$ and we reobtain the groups (i), (ii) or (iii) of the Theorem. 
If $\Lambda^G$ is a sublattice of $2^{27}3^{36}$, then we have an inverse chain of inclusions for the 
stabilisers
\be 
{\rm Stab}(S)\subseteq \ldots \subseteq {\rm Stab}(S_1)\subseteq G={\rm Stab}(S_0)\ ,
\ee
where each stabiliser  group is a subgroup of index at most $2$ in the previous one 
\be 
|{\rm Stab}(S_i):{\rm Stab}(S_{i+1})|\le2\ .
\ee 
Furthermore, whenever $|{\rm Stab}(S_i):{\rm Stab}(S_{i+1})|=2$, we have 
$\rk S_{i+1}>\rk S_i$, and since $\rk S_{2^{27}3^{36}}-\rk \Lambda^G\le 2$, 
we have
\be\label{indexbound} 
\frac{|G|}{|{\rm Stab}(S_{2^{27}3^{36}})|}\le 2^2\ .
\ee Thus, $G$ is the extension of ${\rm Stab}(S_{2^{27}3^{36}})$ by a group $G''$ of order at most 
$4$, and all the possibilities are considered in the case (iv) of the Theorem. This completes the proof
 of the Theorem.

\subsection{Realising all Symmetry Groups of K3}\label{B.4}

Let $G\subset {\rm Aut}(\Lambda)$ be the pointwise stabiliser of a sublattice 
$S\cong \Lambda^G$ of rank $\rk S\ge 4$. In this section, we will prove that, for each such 
$G$, there exists a non-linear $\sigma$-model on K3 having $G$ as its group of symmetries. 
This result is based on the assumption that for every choice of a positive definite $4$-dimensional
subspace  $\Pi\subset \RR^{4,20}$ such that no root $v\in\Gamma^{4,20}$ is orthogonal to $\Pi$, the 
corresponding non-linear $\sigma$-model is well defined.

Let $d+4$, with $d\ge 0$, be the rank of $S$ and let $S^\perp\equiv \Lambda_G$ be 
its orthogonal complement in $\Lambda$. Suppose that $S^\perp(-1)$, obtained from 
$S^\perp$ by changing the sign of its quadratic form, can be primitively embedded into 
$\Gamma^{4,20}$. Then, by a reasoning similar to proposition \ref{th:embed}, we conclude 
that the action of $G$ on $S^\perp(-1)$ can be extended to $\Gamma^{4,20}$, in such a way 
that $(\Gamma^{4,20})_G\cong S^{\perp}(-1)$, {\it i.e.}\ that the sublattice $(\Gamma^{4,20})^G$
invariant under $G$ is the orthogonal complement of $S^\perp(-1)$. Since $(\Gamma^{4,20})^G$ 
has signature $(4,d)$, one can always find a positive definite four dimensional subspace 
$\Pi\subset\RR^{4,20}$ such that $(\Gamma^{4,20})_G=\Gamma^{4,20}\cap\Pi^\perp$. Furthermore, 
$(\Gamma^{4,20})_G$ contains no vectors of norm $-2$, and $G$ is the subgroup of $O(\Gamma^{4,20})$ 
fixing $\Pi$ pointwise. Following the arguments of Section~\ref{symmetries}, $\Pi$ corresponds to a well defined
non-linear $\sigma$-model with symmetry group $G$.

Thus it remains to prove  that a primitive embedding of $S^\perp(-1)$ into $\Gamma^{4,20}$ 
always exists. The `gluing' construction described in Appendix~\ref{s:GinOGamma} shows that, for an even 
lattice with signature $(t_+,t_-)$ and discriminant form $q$  --- we will denote them as a
triple $(t_+,t_-;q)$ ---  the existence of an embedding into some even unimodular lattice with 
signature $(l_+,l_-)$ is equivalent to the existence of an even lattice with signature and discriminant 
form $(l_+-t_+,l_--t_-,-q)$. In \cite[Theorem 1.12.4]{Nikulin},  the following sufficient conditions are proved 
for the existence of such an embedding:
\begin{enumerate}
\item $l_+-l_-\equiv 0\mod 8$
\item $t_+\le l_+$, $t_-\le l_-$ and $t_++t_-\le \frac{1}{2}(l_++l_-)$.
\end{enumerate} 
The lattice $S$ has signature $(4+d,0)$, so that it can be embedded into an even unimodular lattice 
$\Gamma^{8+d,d}\cong E_8\oplus U^d$ of signature $(8+d,d)$, where $U$ is the unimodular lattice of 
signature $(1,1)$. Let $S'$ be the orthogonal complement of $S$ in $\Gamma^{8+d,d}$. If $q$ 
denotes the discriminant form of $S$, then $S'$, $S^\perp$ and $S^\perp(-1)$ have signature and 
discriminant form $(4,d;-q)$, $(20-d,0;-q)$ and $(0,20-d;q)$, respectively. By comparing the signatures 
and discriminant forms of $S^\perp(-1)$ and $S'$, we conclude that these two lattices can be `glued' 
together to form the even unimodular lattice $\Gamma^{4,20}$.

\section{Gepner Models}\label{s:Gepner}

Here we collect, following \cite{Brunner:2005fv}, our conventions for the description of $\N=2$ minimal 
models and Gepner models.

\subsection{$\N=2$ Minimal Models and Gepner Models at $c=6$} \label{App:GepnerModel}

The $\N=2$ minimal model at  level $k$ has central charge $c=\tfrac{3k}{k+2}$, and can 
be described in terms of the coset 
\be 
\frac{\mathfrak{su}(2)_{k+2}\oplus \mathfrak{u}(1)_{4}}{\mathfrak{u}(1)_{2k+4}} 
\ee
that captures  the bosonic subalgebra of the $\N=2$ superconformal algebra. 
The coset representations are labeled by
\be (l,m,s)\ ,\qquad l=0,\ldots,k,\quad m\in \ZZ_{2k+4}\ ,\quad s\in\ZZ_4\ ,
\ee
subject to the condition
\be l+m+s=0\mod 2\ ,
\ee
and with the field identification
\be (l,m,s)\sim (k-l,m+k+2,s+2)\ .
\ee
We denote by $[l,m,s]$ the class corresponding to $(l,m,s)$. In terms of the $\N=2$ algebra,
the irreducible representations are of the form $\Hh_{(l,m,s)}\oplus \Hh_{(l,m,s+2)}$ 
with $s$ even (odd) for the NS (R) sector, since the fermionic generators of the $\N=2$ algebra
map  $\Hh_{(l,m,s)}$ to $\Hh_{(l,m,s+2)}$. The conformal weight and
$U(1)$-charge of the ground state in the $(l,m,s)$ sector are given by
\be h_{l,m,s}=\frac{l(l+2)-m^2}{4(k+2)}+\frac{s^2}{8}\mod\ZZ\ ,\qquad q_{l,m,s}
=\frac{m}{k+2}-\frac{s}{2}\mod 2\ZZ\ .
\ee
The character of the $(l,m,s)$ coset representation equals (with $q=e^{2\pi i\tau}$, $y=e^{2\pi i z}$)
\be \chi_{[l,m,s]}(\tau,z)=\Tr_{{\cal H}_{[l,m,s]}} (q^{L_0-\frac{c}{24}} y^{J_0}) = 
\sum_{j\in \ZZ} c^l_{m+4j-s}(\tau)
q^{\frac{k+2}{2k}(\frac{m}{k+2}-\frac{s}{2}+2j)^2}y^{\frac{m}{k+2}-\frac{s}{2}+2j}\ .
\ee
Here, $c^l_m(\tau)$ can be obtained from the identity
\be \sum_{m\in \ZZ/2k\ZZ} c^l_m(\tau)  \theta_{m,k}(\tau,z)=
\frac{\theta_{l+1,k+2}(\tau,z)-\theta_{-l-1,k+2}(\tau,z)}{\theta_{1,k+2}(\tau,z)-\theta_{-1,k+2}(\tau,z)}\ ,
\ee
where $\theta_{m,k}(\tau,z)$ is the $\mathfrak{su}(2)$ theta function
\be \theta_{m,k}(\tau,z)=\sum_{n\in\ZZ}q^{k(n+\frac{m}{2k})^2}y^{k(n+\frac{m}{2k})}
\ee
with $m\in \ZZ/2k\ZZ$. For the calculation of the elliptic genus we are interested in
the trace with the insertion of $(-1)^F$; for the $\N=2$ representation corresponding
to $(l,m,s)$, this leads to 
\be
I^l_m(\tau,z)=\chi_{[l,m,s]}(\tau,z)-\chi_{[l,m,s+2]}(\tau,z) \ ,\label{CharacterN2rep}
\ee
where $s=0$ in the NS sector and $s=1$ in the Ramond sector.
\medskip

A Gepner model at $c=6$ is defined as a $\ZZ_H$ orbifold of a tensor product
$(k_1)\cdots(k_r)$ of $\N=2$ minimal models, where 
\be
\sum_{i=1}^r \frac{3k_i}{k_i+2}=6\ , \qquad \hbox{and} \qquad H=\lcm\{k_i+2\}\ .
\ee
The spectrum of the orbifold theory is given by
\be
\bigoplus_{\substack{n\in\ZZ_H\\ [l_i,m_i,s_i]}}
\bigotimes_{i=1}^r \Hh_{[l_i,m_i+n,s_i]}\otimes \bar \Hh_{[l_i,m_i-n,\bar s_i]}\ ,
\ee
and the sum is over the sectors $[l_i,m_i,s_i]$ satisfying the orbifold conditions
\begin{align} &\sum_{i=1}^{r}\frac{m_i}{k_i+2}\in\ZZ &\text{$r$ even}\ ,\\
&\sum_{i=1}^{r}\frac{m_i}{k_i+2}-\frac{s_1}{2}\in\ZZ &\text{$r$ odd}\ ,
\end{align}
together with the spin alignment condition
\be s_i-s_j\in 2\ZZ\ ,\qquad i,j=1,\ldots,r\ .
\ee
We denote a state transforming in the $\Hh_{l,m,s}\otimes \bar\Hh_{l,\bar m,\bar s}$ coset
representations by
\be \Phi^l_{m,s;\bar m,\bar s}=\phi_{l,m,s}\otimes\bar \phi_{l,\bar m,\bar s}\ .\label{GepnerStates}
\ee
Imposing the various constraints, it is then clear that the elliptic genus of these Gepner models 
is given by \cite{Kawai:1993jk}
\begin{align}
\phi(\tau,z)&=\frac{1}{2^{r-1}\, H}\sum_{a,b=0}^{H}\prod_{i=1}^r
\sum_{l_i=0}^{k_i} \sum_{m_i= -k_i-1}^{k_i+2}\, 
e^{\frac{2\pi i(m_i+a)b}{k_i+2}}
I^{l_i}_{m_i}(\tau,z)I^{l_i}_{m_i+2a}(\bar{\tau},0)\ .\label{EllGenGepner}
\end{align}
Using 
\begin{align}
I^l_m(\tau,0)=\delta_{m,l+1}-\delta_{m,-l-1}
\end{align}
we can directly evaluate (\ref{EllGenGepner}). For the two cases considered in this paper,
namely $(1)^6$ and $(2)^4$, we then find indeed (\ref{K3eg}).

\subsection{D-branes in Gepner Models}\label{s:Dbranes}

A-type (B-type) D-branes satisfy the gluing conditions
\begin{equation}
\begin{array}{llll}
(L_n-\bar L_{-n})|\!|{\rm A}\rrangle=0 \quad
&(J_n-\bar J_{-n})|\!|{\rm A}\rrangle=0 \quad 
&(G^{\pm}_r+i\eta\bar G^{\pm}_{-r})|\!|{\rm A}\rrangle=0  \quad
& \text{(A-type)} \label{Atype}\\
(L_n-\bar L_{-n})|\!|{\rm B}\rrangle=0
&(J_n+\bar J_{-n})|\!|{\rm B}\rrangle=0
& (G^{\pm}_r+i\eta\bar G^{\pm}_{-r})|\!|{\rm B}\rrangle=0 
&\text{(B-type)}\ ,
\end{array}
\end{equation}
where $\eta=\pm 1$. In our construction we shall always consider tensor product A-type D-branes, as well
as permutation B-type D-branes. The former are described by the boundary states
\cite{Recknagel:1997sb}
\begin{align}
|\!|L_i,M_i,S_i\rrangle_{{\rm A},s}=&\N e^{-\frac{-\pi i}{2}s\sum_iS_i}
\sum_{l_i=0}^{k_i}\sum_{m_i=0}^{k_i+1}\sum_{\nu_i=0}^1
\prod_i\Bigl(\frac{1+(-1)^{l_i+m_i+s}}{2}\Bigr)
\Bigl(\frac{1}{H}\sum_{t\in\ZZ_H}e^{2\pi it\sum_i\frac{m_i}{k_i+2}}\Bigr)\nonumber\\
& \cdot (-1)^{\sum_iS_i\nu_i}e^{\pi i\sum_i M_i\frac{m_i}{k_i+2}}
\prod_i\frac{S_{L_il_i}}{\sqrt{S_{0l_i}}}|l_i,m_i,s+2\nu_i\rrangle_{\rm A}\ ,
\end{align}
where $H=\lcm\{k_i+2\}$, $\N=\sqrt{H}\prod_i\bigl(\frac{2}{k_i+2}\bigr)^{1/4}$ and
\be 
L_i=0,\ldots, k_i \ , \quad M_i\in\ZZ_{2k_i+4} \ , \quad S_i\in\ZZ_4 \ , \qquad
L_i+M_i+S_i\text{ even}\ , \quad S_i+S_j\text{ even.} \nonumber
\ee
Here $s=0,1$ labels the NSNS or the RR closed string sector, respectively,  and
\be
S_{L_il_i}=\sqrt{\frac{2}{k_i+2}}\sin\Bigl(\pi\frac{(L_i+1)(l_i+1)}{k_i+2}\Bigr)
\ee
is the $S$-matrix of the $\mathfrak{su}(2)$ affine algebra at level $k_i$. Furthermore
\be
|l_i,m_i,s+2\nu_i\rrangle_{\rm A}\in \otimes_i\Hh_{[l_i,m_i,s+2\nu_i]}
\otimes\bar\Hh_{[l_i,m_i,s+2\nu_i]}
\ee
is the Ishibashi state satisfying the A-type gluing conditions (\ref{Atype}) for each minimal
model factor separately.  The left-moving world-sheet fermion number operator acts on this boundary 
state as
\be
(-1)^{F_L}\|L_i,M_i,S_i\rrangle_{{\rm A},s}=\|L_i,M_i+1,S_i+1\rrangle_{{\rm A},s}\ .
\ee
The overlap between two such boundary states is given by
\begin{align} {}_{{\rm A},s}\llangle L_i',M_i',S_i'|\!|
&q^{\frac{1}{2}(L_0+\bar L_0)-\frac{c}{24}}(-1)^{F_L} |\!|L_i,M_i,S_i\rrangle_{{\rm A},s}\\
&=\sum_{t\in\ZZ_H}\sum_{(l_i,m_i,s_i)}e^{-\frac{\pi i s}{2}
\sum_i(S_i+1-S_i'+ s_i)}\prod_i\delta^{(2)}(S_i+1-S_i'+s_i)\notag\\
 &\qquad \cdot \prod_i
 \Bigl[\delta^{(k_i+2)}\Bigl(\frac{M_i+1-M_i'+ m_i}{2}+t\Bigr)\,
 N_{L_i l_i}^{L_i'}\, \chi_{[l_i,m_i,s_i]}(\tilde q)\Bigr] \ ,\notag
\end{align}
where $\tilde{q}$ is the open string parameter.
\medskip

Given a permutation $\pi\in S_r$, a permutation brane $|\!|{\rm B}\rrangle^\pi$ satisfies the
$\pi$-twisted B-type boundary conditions
\begin{align}\label{Bpbound}
(L_n^{(i)}-\bar L_{-n}^{(\pi(i))})\|B\rrangle^\pi=0\ , \qquad 
(J_n^{(i)}+\bar J_{-n}^{\pi(i)})\|B\rrangle^\pi=0\ ,
\end{align}
and similarly for the fermionic gluing conditions.
For example, for the $(1)^6$ model we consider $\pi\in S_6$ with cycle decomposition $(12)(34)(56)$.
The corresponding permutation brane is defined as \cite{Recknagel:2002qq}
\begin{align}
|\!|L_1,&L_2,L_3,M_1,M_2,M_3,\hat M,S_i\rrangle^{\pi}_{{\rm B},s} \notag \\
&=\frac{1}{2^3}\frac{1}{\sqrt{3}}\sum_{n\in\ZZ_3}e^{-\frac{\pi in\hat M}{3}}\sum_{l_1,l_2,l_3=0}^1\, 
\sum_{m_1,m_2,m_3=0}^5\, \sum_{\nu_1,\ldots,\nu_6\in\ZZ_2}
\prod_{i}\Bigl(\frac{1+(-1)^{l_i+m_i+n+s}}{2}\Bigr) \notag \\
&\quad \cdot\prod_i
\frac{S_{L_il_i}}{S_{0l_i}}e^{i \frac{\pi}{3} \sum_i^3 m_iM_i}(-1)^{\sum_i^6 S_i\nu_i}
e^{- i \frac{s \pi}{2} \sum_i^6S_i} |[l_i,m_i+n,s+2\nu_i]\rrangle^{\pi,g^n}_{{\rm B}}\ ,  \label{16Bper}
\end{align}
where $\hat M\in\ZZ_6$, while $L_i+M_i$, $\hat M+\sum_iM_i$ and $S_i+S_j$ are even, so that the
boundary state is invariant under $n\mapsto n+3$, $m_i\mapsto m_i+3, i=1,\ldots, 3$. Here,
$|[l_i,m_i+n,s_i]\rrangle^{\pi,g^n}_{{\rm B}}$ is the Ishibashi state in the $n^{\rm th}$ twisted sector
\be
\bigotimes_{i=1}^3
\Bigl( \left(\Hh_{[l_i,m_i+n,s_{2i-1}]}\otimes \bar\Hh_{[l_i,m_i-n,-s_{2i}]} \right) \otimes \left(
\Hh_{[l_i,-m_i+n,s_{2i}]}\otimes\bar\Hh_{[l_i,-m_i-n,-s_{2i-1}]} \right) \Bigr) \ ,
\ee
that is uniquely characterised (up to normalisation) by 
the B-type boundary conditions \eqref{Bpbound} (together with the corresponding fermionic formulae).
Note that
\be
(-1)^{F_L} |\!|L_i,M_i,\hat M,S_i\rrangle^{\pi}_{{\rm B},s}=
|\!|L_i,M_i,\hat M,S_i+1\rrangle^{\pi}_{{\rm B},s}\ .
\ee
For the case at hand, the overlap between two such branes with $L_i=0$ and $s=1$ (RR sector) 
is
\begin{align}
&{}^{\pi}_{{\rm B},s=1}\llangle 0,0,0,M_1',M_2',M_3',\hat M',0|\!|
q^{\frac{1}{2}(L_0+\bar L_0)-\frac{c}{24}}(-1)^{F_L}
|\!|0,0,0,M_1,M_2,M_3,\hat M,0\rrangle^{\pi}_{{\rm B},s=1}\notag\\
&=\sum_{ m_i=\pm 1}
e^{-\frac{\pi i }{2}\sum_i^6(1+ m_i)}
\delta^{(3)}\Bigl(\frac{\hat M'-\hat M+\sum_i m_i}{2}\Bigr)
\prod_{i=1}^3\delta^{(3)}\Bigl(\frac{M_i-M_i'+ m_{2i-1}- m_{2i}}{2}\Bigr) \ .
\end{align}
\medskip

In order to determine the relative intersection number between the A-type and the B-type D-branes,
we also need to understand the overlap between the different types of branes. 
For the $(1)^6$ model, the only non-vanishing overlaps are between Ishibashi states in the
$n=0$ sectors of the form
\be
\bigotimes_{i=1}^3 (\Hh_{[l_i,m_i,s_i]}\otimes
\bar\Hh_{[l_i,m_i,s_i]}\otimes\Hh_{[l_i,-m_i,-s_{i}]}\otimes\bar\Hh_{[l_i,-m_i,-s_i]})\ .
\ee
Taking into account the implicit relative phase between the Ishibashi states as in \cite{Brunner:2005fv},
one finds that the overlap between the relevant Ishibashi states is
\be {}_{\rm A}\llangle l_i,m_i,s_i| q^{\frac{1}{2}(L_0+\bar L_0)-\frac{c}{24}}
|[l_i,m_i+n,s_i]\rrangle^{\pi}_{\rm B}
=\prod_{i=1}^3e^{\frac{\pi i}{3}m_i-\frac{\pi i}{2}s_i}\chi_{[l_i,m_i,s_i]}(q^2)\ ,
\ee
and hence the overlap between the relevant boundary states is given by
\begin{align} {}_{{\rm A},s=1}&\llangle0,M_i',0|\!| q^{\frac{1}{2}(L_0+\bar L_0)-\frac{c}{24}}(-1)^{F_L}
|\!|0,0,0,M_1,M_2,M_3,\hat M,0\rrangle^{\pi}_{{\rm B},s=1}\notag\\
=&\sum_{m_1, m_2, m_3=\pm 1}-e^{-\frac{\pi i}{2}(1+\sum_i m_i)}\prod_i^3
\delta^{(3)}\Bigl(\frac{M_i-M_{2i-1}'+M_{2i}'+ m_i+1}{2}\Bigr)\ .
\end{align}

A similar computation can be done for the $(2)^4$ model. The B-type permutation brane for 
$\pi\in S_4$ with cycle decomposition $(12)(34)$ is
\begin{align}
|\!|L_1,&L_2,M_1,M_2,\hat M,S_i\rrangle^{\pi}_{{\rm B},s} \notag \\
&=\frac{1}{2^2}\frac{1}{\sqrt{4}}\sum_{n\in\ZZ_4}e^{-i \frac{\pi n}{4} \hat M}\sum_{l_1,l_2=0}^2\, 
\sum_{m_1,m_2=0}^7\, \sum_{\nu_1,\ldots,\nu_4=0}^1
\prod_{i}\Bigl(\frac{1+(-1)^{l_i+m_i+n+s}}{2}\Bigr) \label{24Bper} \\
& \quad\cdot \prod_i
\frac{S_{L_il_i}}{S_{0l_i}}e^{i \frac{\pi}{4} \sum_im_iM_i}(-1)^{\sum_i^4 S_i\nu_i}
e^{- i \frac{\pi s}{2} \sum_i^4S_i}|[l_i,m_i+n,s+2\nu_i]\rrangle^{\pi,g^n}_{{\rm B}}\ , \notag 
\end{align}
where $\hat M\in\ZZ_8$, and $L_i+M_i$, $\hat M+\sum_iM_i$ and $S_i+S_j$ are all even. The
left-moving  fermionic number operator acts now by
\be
(-1)^{F_L} |\!|L_1,L_2,M_1,M_2,\hat M,S_i\rrangle^{\pi}_{{\rm B},s}=
|\!|L_1,L_2,M_1,M_2,\hat M+4,S_i+1\rrangle^{\pi}_{{\rm B},s}\ ,
\ee
and the overlap between two branes with $L_1,L_2=0$ is
\begin{align}
&{}^{\pi}_{{\rm B},s=1}\llangle 0,0,M_1',M_2',\hat M',0|\!|
q^{\frac{1}{2}(L_0+\bar L_0)-\frac{c}{24}}(-1)^{F_L}
|\!|0,0,M_1,M_2,\hat M,0\rrangle^{\pi}_{{\rm B},s=1} \nonumber \\
&=\frac{1}{4} \sum_{(l_i, m_i,s_i)}\delta_{ l_1 l_2}\delta_{ l_3 l_4}\,
\prod_{i=1}^{4}\delta^{(2)}(s_i+1)\, 
e^{-\frac{\pi i }{2}\sum_i^4(1+s_i)}
\delta^{(4)}\Bigl(\frac{\hat M'-\hat M+4+\sum_i^4 m_i}{2}\Bigr) \nonumber\\
&\hspace{2.5cm}\cdot 
\prod_{i=1}^2\delta^{(4)}\Bigl(\frac{M_i-M_i'+m_{2i-1}- m_{2i}}{2}\Bigr)\,
\chi_{[l_i m_i s_i]}(\tilde{q}) \notag\\
&=\sum_{\substack{ m_1, m_2\in\{1,2,3\}\\ s_1, s_2=\pm 1}}e^{-\frac{\pi i }{2}(2+s_1+s_2)}\,
\delta^{(4)}\Bigl(\frac{\hat M'-\hat M+4+\sum_i^2 m_i( s_i+1)}{2}\Bigr) \nonumber\\
&\hspace{2.5cm}\cdot \prod_{i=1}^2\delta^{(4)}\Bigl(\frac{M_i-M_i'+ m_{i}(s_i-1)}{2}\Bigr)\ .
\end{align} 
Finally, the RR-overlap between
A-type and permutation B-type branes in this model is given by
\begin{align} {}_{{\rm A},s=1}&\llangle0,M_i',0|\!| q^{\frac{1}{2}(L_0+\bar L_0)-\frac{c}{24}}(-1)^{F_L}
|\!|0,0,0,M_1,M_2,\hat M,0\rrangle^{\pi}_{{\rm B},s=1}\notag\\
=&\sum_{m_1, m_2=\pm 1}e^{-\frac{\pi i}{2}(2+\sum_i m_i)}
 \delta^{(4)}\Bigl(\frac{M_1-M_{1}'+M_{2}'+ m_1+1}{2}\Bigr)\nonumber\\
&\hspace{2.5cm}
 \cdot\delta^{(4)}\Bigl(\frac{M_2-M_{3}'+M_{4}'+ m_2+1}{2}\Bigr)\ .
\end{align}


\begin{thebibliography}{99}

\bibitem{EOT} 
T.~Eguchi, H.~Ooguri and Y.~Tachikawa,
{\it Notes on the K3 surface and the Mathieu group $M_{24}$},
Exper.\ Math.\ {\bf 20} (2011) 91
{\tt [arXiv:1004.0956 [hep-th]]}.

\bibitem{Gannon} 
T.~Gannon, 
{\it Moonshine beyond the Monster: The bridge connecting algebra, modular forms and physics},
Cambridge University Press (2006).

\bibitem{Cheng:2010pq}
M.C.N.~Cheng,
{\it K3 Surfaces, N=4 dyons, and the Mathieu group $M_{24}$},
Commun.\ Number\ Theory\ Phys. {\bf 4}  (2010) 623
{\tt [arXiv:1005.5415 [hep-th]]}.

\bibitem{Gaberdiel:2010ch}
M.R.~Gaberdiel, S.~Hohenegger and R.~Volpato,
{\it Mathieu twining characters for K3},
JHEP {\bf 1009} (2010) 058
{\tt [arXiv:1006.0221 [hep-th]]}.

\bibitem{Gaberdiel:2010ca}
M.R.~Gaberdiel, S.~Hohenegger and R.~Volpato,
{\it Mathieu Moonshine in the elliptic genus of K3},
JHEP {\bf 1010} (2010) 062
{\tt [arXiv:1008.3778 [hep-th]]}.

\bibitem{Eguchi:2010fg}
T.~Eguchi and K.~Hikami,
{\it Note on twisted elliptic genus of K3 surface},
Phys.\ Lett.\  B {\bf 694} (2011) 446
{\tt [arXiv:1008.4924 [hep-th]]}.

\bibitem{Govindarajan:2010cf}
S.~Govindarajan,
{\it Brewing moonshine for Mathieu},
{\tt arXiv:1012.5732 [math.NT]}.

\bibitem{Govindarajan:2010fu}
S.~Govindarajan,
{\it BKM Lie superalgebras from counting twisted CHL dyons},
{\tt arXiv:1006.3472 [hep-th]}.

\bibitem{Govindarajan:2009qt}
S.~Govindarajan and K.~Gopala Krishna,
{\it BKM Lie superalgebras from dyon spectra in Z(N) CHL orbifolds for composite N},
JHEP {\bf 1005} (2010)  014
{\tt [arXiv:0907.1410 [hep-th]]}.

\bibitem{Mukai}
S.~Mukai,
{\it Finite groups of automorphisms of $K3$ surfaces and the Mathieu group},
Invent.\ Math. {\bf 94} (1988) 183.

\bibitem{Kondo}
S.~Kondo,
{\it Niemeier lattices, Mathieu groups and finite groups of symplectic automorphisms
of K3 surfaces},
Duke Math.\ Journal {\bf 92} (1998) 593, appendix by S.~Mukai.

\bibitem{Taormina:2010pf}
A.~Taormina and K.~Wendland,
{\it The symmetries of the tetrahedral Kummer surface in the Mathieu group $M_{24}$},
{\tt arXiv:1008.0954 [hep-th]}.

\bibitem{Aspinwall:1996mn}
P.S.~Aspinwall,
{\it K3 surfaces and string duality},
{\tt arXiv:hep-th/9611137}.

\bibitem{NahmWend}
W.~Nahm and K.~Wendland,
{\it A hiker's guide to K3: Aspects of N = (4,4) superconformal field theory  with central charge c = 6},
Commun.\ Math.\ Phys.\  {\bf 216} (2001) 85
{\tt [arXiv:hep-th/9912067]}.

\bibitem{Brunner:1999jq}
I.~Brunner, M.R.~Douglas, A.E.~Lawrence and C.~Romelsberger,
{\it D-branes on the quintic},
JHEP {\bf 0008} (2000) 015
{\tt [arXiv:hep-th/9906200]}.

\bibitem{Banks:1988yz}
T.~Banks and L.J.~Dixon,
{\it Constraints on string vacua with space-time supersymmetry},
Nucl.\ Phys.\  B {\bf 307} (1988) 93.

\bibitem{Aspinwall:1995xy}
P.S.~Aspinwall,
{\it Enhanced gauge symmetries and Calabi-Yau threefolds},
Phys.\ Lett.\  B {\bf 371} (1996) 231
{\tt [arXiv:hep-th/9511171]}.

\bibitem{Strominger:1995cz}
A.~Strominger,
{\it Massless black holes and conifolds in string theory},
Nucl.\ Phys.\  B {\bf 451} (1995)  96
{\tt [arXiv:hep-th/9504090]}.

\bibitem{Gaberdiel:2000jr}
M.R.~Gaberdiel,
{\it Lectures on non-BPS Dirichlet branes},
Class.\ Quant.\ Grav.\  {\bf 17} (2000) 3483
{\tt [arXiv:hep-th/0005029]}.

\bibitem{Bergman:1999kq}
O.~Bergman and M.R.~Gaberdiel,
{\it Non-BPS states in heterotic type IIA duality},
JHEP {\bf 9903} (1999) 013
{\tt [arXiv:hep-th/9901014]}.

\bibitem{Brunner:2006tc}
I.~Brunner, M.R.~Gaberdiel and C.A.~Keller,
{\it Matrix factorisations and D-branes on K3},
JHEP {\bf 0606} (2006) 015
{\tt [arXiv:hep-th/0603196]}.

\bibitem{Conway}
J.H.~Conway and N.J.A.~Sloane,
{\it Sphere packings, lattices and groups},
Grundlehren der Mathematischen Wissenschaften
{\bf 290} 3rd edition, Springer-Verlag, New York (1999).

\bibitem{Curtis}
R.T.~Curtis,
{\it On subgroups of {$\cdot 0$}. {I}. {L}attice stabilizers},
J.\ Algebra {\bf 27} (1973) 549.

\bibitem{Brunner:2009mn}
I.~Brunner, M.R.~Gaberdiel, S.~Hohenegger and C.A.~Keller,
{\it Obstructions and lines of marginal stability from the world-sheet},
JHEP {\bf 0905} (2009) 007 
{\tt [arXiv:0902.3177 [hep-th]]}.

\bibitem{Govindarajan:2011mp}
S.~Govindarajan, D.P.~Jatkar and K.G.~Krishna,
{\it BKM superalgebras from counting dyons in N=4 supersymmetric type II compactifications},
{\tt arXiv:1106.1318 [hep-th]}.

\bibitem{Nikulin}
V.V.~Nikulin,
{\it Integer symmetric bilinear forms and some of their geometric applications},
Izv.\ Akad.\ Nauk SSSR Ser.\ Mat.\   {\bf 43} (1979) 111.

\bibitem{Atlas}
J.H.~Conway, R.T.~Curtis, S.P.~Norton, R.A.~Parker and R.A.~Wilson,
{\it Atlas of finite groups},
Oxford University Press (1985).

\bibitem{Allcock}
D.~Allcock,
{\it Orbits in the Leech lattice,}
Exper.\ Math. {\bf 14}(4) (2005) 491.

\bibitem{Brunner:2005fv}
I.~Brunner and M.R.~Gaberdiel,
{\it Matrix factorisations and permutation branes},
JHEP {\bf 0507} (2005)  012
{\tt [arXiv:hep-th/0503207]}.

\bibitem{Kawai:1993jk}
T.~Kawai, Y.~Yamada and S.K.~Yang,
{\it Elliptic genera and N=2 superconformal field theory},
Nucl.\ Phys.\  B {\bf 414} (1994) 191
{\tt [arXiv:hep-th/9306096]}.

\bibitem{Recknagel:1997sb}
A.~Recknagel and V.~Schomerus,
{\it D-branes in Gepner models},
Nucl.\ Phys.\  B {\bf 531} (1998) 185 
{\tt [arXiv:hep-th/9712186]}.

\bibitem{Recknagel:2002qq}
A.~Recknagel,
{\it Permutation branes},
JHEP {\bf 0304} (2003)  041
{\tt  [arXiv:hep-th/0208119]}.

\end{thebibliography}
\end{document}